\documentclass[nonacm,sigplan,screen]{acmart}
\input{preamble.tex}

\setcopyright{none}

\theoremstyle{remark}
\newtheorem*{notations}{Notations}

\begin{document}

\title[Quantum Expectation Transformers for Cost Analysis]{Quantum Expectation Transformers for \\ Cost Analysis}


\author{Martin Avanzini}
\orcid{nnnn-nnnn-nnnn-nnnn}             
\affiliation{
  \institution{Inria Sophia Antipolis - M\'editerran\'ee}            
  \country{France}                    
}

\author{Georg Moser}
\orcid{nnnn-nnnn-nnnn-nnnn}             
\affiliation{
  \department{Department of Computer Science}             
  \institution{Universit\"at Innsbruck}           
  \city{Innsbruck}
  \country{Austria}                   
}

\author{Romain P\'echoux}
\orcid{nnnn-nnnn-nnnn-nnnn}             
\affiliation{
  \institution{Universit{\'e} de Lorraine, CNRS, Inria, LORIA, F-54000  Nancy}        
  \country{France}                   
}

\author{Simon Perdrix}
\orcid{nnnn-nnnn-nnnn-nnnn}             
\affiliation{
  \institution{Universit{\'e} de Lorraine, CNRS, Inria, LORIA, F-54000 Nancy}        
  \country{France}     
}

\author{Vladimir Zamdzhiev}
\orcid{nnnn-nnnn-nnnn-nnnn}             
\affiliation{
  \institution{Universit{\'e} de Lorraine, CNRS, Inria, LORIA, F-54000 Nancy}        
  \country{France}     
}

\begin{abstract}
  We introduce a new kind of expectation transformer for a mixed
  \emph{classical-quantum programming language}.  Our semantic approach relies on a
  new notion of a cost structure, which we introduce and which can be seen as a
  specialisation of the Kegelspitzen of Keimel and Plotkin.  We show that our
  weakest precondition analysis
  is both sound and adequate with respect to the operational
  semantics of the language.
  Using the induced expectation transformer, we provide formal analysis methods for the
  expected cost analysis and expected value analysis of classical-quantum programs.
  We illustrate the usefulness of our techniques by computing the expected cost of
  several well-known quantum algorithms and protocols, such as coin tossing,
  repeat until success, entangled state preparation, and quantum walks.
\end{abstract}

\ccsdesc[100]{Theory of computation~Complexity theory and logic}

\keywords{complexity analysis, quantum programming, expectation transformer, formal semantics}

\maketitle

\section{Introduction}
Quantum computation is a promising and emerging computational paradigm which
can efficiently solve problems considered to be intractable on classical
computers \cite{shor,hhl}. However, the unintuitive nature of quantum
mechanics poses interesting and challenging questions for the design and
analysis of quantum programming languages. Indeed, the quantum program dynamics
are considerably more complicated compared to the behaviour of classical
probabilistic programs. Therefore, formal reasoning about quantum programs
requires the development of novel methods and tools.

An important open problem is to compute the expected resource usage of quantum
programs. For example, this may be used to determine: (1) the expected runtime;
(2) the expected number of quantum gates; or (3) the amount of quantum
resources (in an application-specific sense) required by quantum programs, etc.
The difficulty of this problem, which is undecidable, requires using elaborate
methods to solve it whenever possible.
These methods for estimating resource usage must be compositional, systematic,
and, preferably, tractable; this excludes de facto any direct use of the operational semantics.

We address this open problem by establishing a weakest precondition reasoning in the
form of a \emph{quantum expectation transformer}, named \emph{$\mathtt{qet}$-calculus}, that
is rich enough to recover earlier wp-calculi in the context of classical programs as
well as denotational semantics for quantum programs. Further, the calculus appears to
be the right foundation for subsequent automation of the method, which however, is left
for future work.
The exact solution of the expected cost problem
can be recovered via this calculus, and furthermore, our method may also
be used to find \emph{approximate} solutions by identifying suitable upper
bounds. Therefore, our method provides a basis for attacking and ameliorating
this undecidable problem in a systematic and compositional way.

\subsection{Our Contributions}

As a first step towards achieving our main objective, we introduce a new
domain-theoretic notion, called a \emph{cost structure} (\Cref{sec:kegel}).
It is based on Kegelspitzen \cite{kegelspitzen}, which are dcpo's
(directed-complete partial orders) equipped with a suitable convex structure
that may be used to reason about the semantics of probabilistic
\cite{rennela-kegel,commutative-monads} and quantum programming languages
\cite{popl22}.  A cost structure is then a pair $(\CSd, \CSp)$ of a Kegelspitze
$\CSd$ together with a \emph{cost addition} operation $\CSp$ that allows us to
model resource consumption in a coherent way.

We introduce a mixed classical-quantum programming language on which we formally
define the expected cost and the expected value of programs. 
Our programming language (\Cref{s:l}) supports conditional branching,
while loops, the usual quantum primitives (including quantum measurements),
classical data, and a special statement for resource consumption.
To seamlessly model the combination of
cost primitives and probabilistic choice --- induced by quantum measurements ---
we define the operational semantics of our language as a probabilistic abstract reduction
system~\cite{BG:RTA:05}, whose reduction rules are annotated by costs \cite{AMS20}.

In \Cref{sec:qet}, we introduce the aformentioned \textsf{qet}-calculus,
which can be seen as a generalisation of previous work on predicate transformers and probabilistic expectation transformers.
For a given cost structure $(\CSd, \CSp)$, our quantum expectation transformer is a semantic function
  \[
    \qet{\cdot}{\cdot} : \Progs \to \CSd^\memst \to \CSd^\memst
  \]
which maps programs and \emph{expectations} (functions from quantum program
states to a cost structure) to expectations, in a continuation passing style.
We prove that our semantics enjoys nice algebraic and domain-theoretic
properties (\secref{sub:prop}) and that it is sound and adequate with respect
to the operational semantics (\secref{sub:adequacy}).  As a consequence, we
prove that the expected cost of a program in our mixed classical-quantum
language (as defined via the operational semantics) is precisely recovered
by using our quantum expectation transformer (Corollary \ref{c:ect-soundness}).
Furthermore, because our semantics is defined in a suitable level of
generality, by choosing an appropriate cost structure $(\CSd, \CSp)$, we show how
a strongly adequate quantum denotational semantics may be defined as a
\emph{special case} (\secref{sub:denotational-semantics}),
which highlights important connections between our approach and denotational
semantics of probabilistic and quantum programming languages.

The usefulness of our methods are illustrated through a running example that
performs (unbounded) coin tossing using quantum resources. More useful and
complicated quantum programs are analysed in \Cref{s:ex}, where we show
how we can determine the expected cost of these programs using our quantum
expectation transformer method.

\subsection{Related Work}
\paragraph{Classical and probabilistic programs}
\emph{Predicate transformer} semantics were introduced in the seminal works of
~\cite{D76} and ~\cite{K85} as a method for reasoning about the semantics of
imperative programs. Predicate transformers map each program statement to a
function between two predicates on the state space of the program.
Consequently, their semantics can be viewed as a reformulation of Floyd–Hoare
logic~\cite{hoare1969axiomatic},
since they transform postconditions (the output
predicate) to preconditions (the input predicate).
%
This methodology has been extended to probabilistic programs by replacing predicates
with expectations, leading to the notion of \emph{expectation transformers} (see~\cite{MM05})
and the development of weakest pre-expectation semantics~\cite{GKM14}. Expectation transformers
have been used to reason about expected values~\cite{KK17},
but also runtimes~\cite{KKMO16}, and costs \cite{AMS20,ABL21,NCH18}.

\paragraph{Quantum programs}
The articles~\cite{OD20} and~\cite{LZY19} present two first attempts to adapt
expectation transformers to the runtime analysis of quantum programs.
\cite{OD20} discusses the interest of adapting the method to the quantum case
through a running example. However, no correctness results (soundness or
adequacy) are proved.  \cite{LZY19} defines a notion of expected runtime
transformer that is neither compositional, nor denotational, because its
definition depends on the asymptotic (i.e., limiting) behaviour of the
operational semantics, which is problematic and undesirable, as we discussed
above. Our paper overcomes all these drawbacks by defining a compositional
and denotational notion of quantum expectation transformers, that is
completely independent of the operational semantics.  Moreover, quantum
expectation transformers are not restricted only to runtimes, and we also
establish the necessary correctness results (soundness and adequacy) with
respect to the operational semantics.  Furthermore, since our language includes
classical data (the other papers do not) and we can easily represent discrete
probabilistic choice, our quantum expectation transformers can be seen as a
proper generalisation of the predicate and expectation transformers discussed
above, which is another advantage of our approach.

\section{Kegelspitzen and Cost Structures}
\label{sec:kegel}

We begin by defining a notion of \emph{cost structure} based on the domain-theoretic and convex
structure of Kegelspitzen. This is used in later sections by our quantum
expectation transformers in order to formalise the semantics.

Kegelspitzen~\cite{kegelspitzen} are dcpo's (directed complete partial orders) that enjoy a convex structure.
We define our quantum expectation transformer by making use of Kegelspitzen, but for simplicity, we define Kegelspitzen using $\omega$-cpo's ($\omega$-complete partial orders), instead of dcpo's, because the former notion is more familiar to most readers.

\begin{definition}
  An \emph{$\omega$-chain} in a partial order $(X, \leq)$ is a countable increasing sequence of elements of $X,$ i.e., a sequence $(x_i)_{i \in \mathbb N}$, such that $x_i \leq x_j$ for any $i \leq j$.
  An \emph{$\omega$-cpo ($\omega$-complete partial order)} is a partial order $(X, \leq)$, such that every $\omega$-chain in $X$ has a supremum (least upper bound) within $X$.
  A monotone function $f : X \to Y$ between two $\omega$-cpos is \emph{$\omega$-continuous} if it preserves suprema of $\omega$-chains.
\end{definition}

Next, we recall barycentric algebras, which allows us to take convex combinations of elements in a coherent way.
\begin{definition}[\cite{kegelspitzen}]
A \emph{barycentric algebra} is a set $A$ equipped with binary operations $a+_{r}b,$ one for every real number $r\in [0, 1],$ such that for all $a, b, c \in A$ and $r, p\in [0, 1)$, the following equalities hold:
\begin{align*}
a+_{1} b &= a; & a+_{r}b  &= b+_{1-r}a; \\
a+_{r} a &= a; &(a+_{p}b)+_{r}c &= a+_{pr}(b+_{\frac{r-pr}{1-pr}} c).
\end{align*}
\end{definition}

Next we introduce pointed barycentric algebras, which allow us to also define scalar multiplication in a natural way.

\begin{definition}[\cite{kegelspitzen}]
A \emph{pointed bary\-centric algebra} is a barycentric algebra~$A$ equipped with a distinguished element $\bot$.
For $a\in A$ and $r\in [0, 1]$, we define scalar multiplication as $r\cdot a \defeq a+_{r} \bot$.
\end{definition}

We can now define an $\omega$-Kegelspitze as a pointed barycentric algebra that respects the order of an $\omega$-cpo.

\begin{definition}
  An $\omega$-\emph{Kegelspitze} is a pointed barycentric algebra $K$ equipped with an $\omega$-complete partial order such that, (1)
  scalar multiplication $(r, a)\mapsto r\cdot a \colon [0,1]\times K\to K$ is $\omega$-continuous in both arguments,
  and (2) for every $r \in [0,1]$ the functions $(a, b)\mapsto a+_{r}b\colon K\times K\to K$ are $\omega$-continuous in both arguments.
\end{definition}

For brevity, we will refer to $\omega$-Kegelspitzen simply as Kegelspitzen. In fact, all $\omega$-Kegelspitzen we consider in this paper are also Kegelspitzen in the sense of \cite{kegelspitzen} (i.e., as dcpo's), so this should not lead to confusion.
We note that, in every Kegelspitze $K$, scalar multiplication $(r, a)\mapsto r\cdot a = a +_r \bot$
is $\omega$-continuous and therefore monotone in the $r$-component, which implies $\bot
= \bot +_{1} a = a+_{0} \bot = 0\cdot a\leq 1\cdot a = a$ for each $a\in K$.
Therefore, the distinguished element $\bot$ is the least element of $K$.

\begin{example}
  The real unit interval $[0,1]$ is a Kegelspitze in the usual order when we define $a +_r b \eqdef ra + (1-r)b$ and $\bot \eqdef 0.$
  The same assignment can also be used to equip the extended non-negative reals $ \Rext \triangleq \Rpos \cup \{\infty\}$ with the structure of a Kegelspitze.
  Note that the non-negative reals $\mathbb R^+$ is \emph{not} a Kegelspitze, because it lacks an $\omega$-cpo structure.
\end{example}

Next, we consider some Kegelspitzen which are important for
the semantics of quantum programming languages.

\begin{example}
\label{ex:density-matrix}
  A \emph{density matrix} is a positive semi-definite hermitian matrix $A$, such that $\mathrm{tr}(A) = 1.$
  A \emph{subdensity matrix} is a positive semi-definite hermitian matrix $A$, such that $\mathrm{tr}(A) \leq 1.$
  Let $D_n \subseteq \mathbb C^{n \times n}$ be the set of subdensity matrices of dimension $n$.
  Then $D_n$ is an $\omega$-cpo when equipped with the L\"owner order: $A \leq B$ iff $B - A$ is positive semi-definite \cite{qpl}.
  Moreover, $D_n$ has the structure of a Kegelspitze under the assignment $\perp \eqdef \mathbf 0$ and $A +_r B \eqdef rA + (1-r) B.$
\end{example}

Recall that density matrices are used in quantum physics to represent
probabilistic mixtures of pure quantum states. In quantum programming
semantics, we use \emph{sub}density matrices in order to account for the
probability of non-termination.

Kegelspitzen may also be used to define convex sums.

\begin{definition}
  In a Kegelspitze $K$, for $a_{i}\in K, r_{i}\in [0, 1]$ with $\sum_{i=1}^{n}r_{i}\leq 1$, we define the \emph{convex sum} inductively by:
  \[
    \sum_{i=1}^{n}r_{i}a_{i} \defeq
    \begin{cases}
      \bot  &  \text{if } n =0,\\
       a_n  &  \text{if } n>0 \text{ and } r_n =1,\\
      a_{n}+_{r_{n}}(\sum_{i=1}^{n-1} \frac{r_{i}}{1-r_{n}}a_{i}) &  \text{otherwise}.
    \end{cases}
  \]
  In fact, the expression $\sum_{i=1}^{n}r_{i}a_{i}$ is $\omega$-continuous in each $r_{i}$ and $a_{i}$ and the sum is also invariant under index permutation (see \cite{commutative-monads} for more details).
  Countable convex sums may be defined as follows: given $a_i \in K$ and $r_i \in [0,1]$, for $i \in \mathbb N$, with $\sum_{i \in \mathbb N} r_i \leq 1$, let
  $ \sum_{i \in \mathbb N} r_i a_i \defeq \sup_{n \in \mathbb N} \sum_{j=1}^n r_j a_j . $
\end{definition}

We now formalize a notion of \emph{cost structure} for expectation transformers in the context of quantum programs. This can be seen as a Kegelspitze equipped with an operation for injecting a cost --- modeled as a positive real number --- into the Kegelspitze,
which satisfies some coherence conditions with respect to the structure of the Kegelspitze.
\begin{definition}\label{d:cs}
A \emph{cost structure} $\CS = (\CSd, \CSp)$ is a Kegelspitze $\CSd$
equipped with an operation
${\CSp} : \Rext \times \CSd \to \CSd$
that is $\omega$-continuous in both arguments and satisfies the identities
\begin{align}
  \label{eq:CSp}
  0 \CSp s & = s \\
  c \CSp (d \CSp s) & = (c + d) \CSp s \\
  (c_1 \CSp s_1) +_{r} (c_2 \CSp s_2) & = (c_1 +_r c_2) \CSp (s_1 +_{r} s_2)
\end{align}
\end{definition}

\begin{example}
  For any Kegelspitze $\CSd$, we get a cost structure $(\CSd,+_{\mathtt{f}})$ with \emph{forgetful cost addition} defined by $c +_{\mathtt{f}} r \triangleq r$.
  A more representative example is given by the cost structure $(\mathbb \Rext, + ),$ where $+$ is the standard addition in $\Rext.$
\end{example}

\section{Quantum Programming Language}\label{s:l}
In this section we introduce the syntax and operational semantics of our imperative programming language supporting both quantum and classical programming primitives.
\subsection{Syntax}
Let $\Bool$, $\Var$ and $\Qubits$ be three distinct types for Boolean, numerical, and qubit data. We will use variables $\x, \y,\z$ to range over classical variables of type $\K \in \{\Bool,\Var\}$ and we will use $\q,\q_1,\q_2,$ etc., to range over quantum variables of type $\Qubits$.
The syntax of quantum programs is described in Figure~\ref{fig:synt}, where $n$ is a constant in $\mathbb{Z}$, $\ope$ is an operator symbol of arity $ar(\ope) \in \mathbb{N}-\{0\}$, $\overline{\q}$ stands for a sequence of qubit variables $\q_1,\ldots,\q_{ar(\ope)}$, and $\meas(\q)$ represents the standard measurement on qubit $\q$ in the computational basis.
When needed, variables and expressions can be annotated by their type as superscript. 
If $\nexp$ evaluates to a positive integer $c$, the statement $\cons(\nexp)$ consumes $c$ resource units but acts as a no-op otherwise.
That is, we permit only non-negative costs. This restriction is in place to ensure
that the notion of expected cost --- to be defined in a moment --- is well-defined.
\begin{figure*}[!h]
\hrulefill
\\[10pt]
$ \begin{array}{llll}
 \AExp &  \nexp, \nexp_1, \nexp_2
    & \rgl &   \x^\Var \ | \ n  \ | \
     \nexp_1 + \nexp_2   \ | \  \nexp_1 - \nexp_2  \ | \ \nexp_1 \times \nexp_2 \\
\BExp &  \bexp, \bexp_1, \bexp_2
    & \rgl &  \x^\Bool   \ | \ \true \ | \ \false \ | \
      \nexp_1 = \nexp_2  \ | \  \nexp_1 \leq \nexp_2 \ | \ \nexp_1 < \nexp_2 \ | \ \neg \bexp \ | \ \bexp_1 \wedge \bexp_2 \ | \ \bexp_1 \vee \bexp_2 \\
      \Exp &  \e, \e_1, \e_2
    & \rgl &  \nexp \ | \ \bexp \\
\Cmds &  \cmd, \cmd_1, \cmd_2
    & \rgl &  \skp   \ | \  \x^\K \asg \e^\K \ | \      \overline{\q}  \oasg \ope \ | \   \x^\Bool \asg \meas(\q)  \ | \  \cons(\nexp) \\
  &  & &\ | \ \cmd_1 \sep \cmd_2 \ |\   \ifa (\bexp) \{\cmd_1\} \elsea \{\cmd_2\}  \ |\ \while(\bexp)\{ \cmd \}  \\
 \end{array}$
 \\[10pt]
 \hrulefill
 \vspace{-2mm}
 \caption{Syntax of quantum programs.}
\label{fig:synt}
\end{figure*}

Program variables are global. For a given expression or statement $t$, let $\Bool(t)$ (respectively $\Var(t)$, $\Qubits(t)$) be the set of Boolean (resp. numerical, qubit) variables in $t$. 

\begin{example}\label{ex:cointoss}
  Let $\oper{H}$ be the operator symbol representing the Hadamard unitary operation. The program $CT(\q)$
  in \Cref{ex1} performs coin tossing by repeatedly measuring an initial qubit $\q$ (which may be mapped into a superposition state via $\oper{H}$) until the measurement outcome $\false$ occurs.
  This program will be our simple running example throughout the paper.
  Its probability to terminate within $n$ steps depends on the initial state of the qubit $\q$ and the loop consumes $1$ resource for each iteration.
  The overall probability of termination (in any number of steps) is $1$.
  \begin{figure}
    \centering
\begin{lstlisting}[caption={Coin tossing.}, label={ex1}, style=qwhile]
$CT(\q) \triangleq$ $\x \asg \true \sep$
        $\while(\x) \{ $                                 $\qquad\quad\tikzrom{cq-mark1}$
          $\q \oasg \oper{H} \sep$ $\tikzrom{cq-mar1}$
          $\x \asg \meas(\q) \sep$ $\tikzrom{cq-mar3}$   $\qquad\quad\tikzrom{cq-mark3}$
          $\cons(1)$               $\tikzrom{cq-mar2}$
        $\}$                                             $\qquad\quad\tikzrom{cq-mark2}$
\end{lstlisting}
\AddNotecopy{cq-mar1}{cq-mar2}{cq-mar3}{$\cmd_0$}
\AddNotecopy{cq-mark1}{cq-mark2}{cq-mark3}{$\cmd$}
  \end{figure}
\end{example}

\subsection{Operational Semantics}\label{sec:os}
In what follows, we model the dynamics of our language as a probabilistic abstract reduction system \cite{BG:RTA:05} ---
a transition system where reducts are chosen from a probability distribution.
Reductions can then be defined as stochastic processes~\cite{BG:RTA:05}, or equivalently,
as reduction relations over distributions~\cite{ALY20}. We follow the latter approach,
unlike the former it permits us to define a notion of expected cost concisely, without
much technical overhead \cite{AMS20}.

\paragraph*{Probabilistic Abstract Reduction Systems (PARS)}

Let $A$ be a set of \emph{objects}.
A discrete subdistribution $\delta$ over $A$ is a function $\delta : A \to [0,1]$ with countable support that maps an element $a$ of $A$ to a probability $\delta(a)$ such that $\sum_{a \in \supp(\delta)} \delta(a) \leq 1$.
If $\sum_{a \in \supp(\delta)} \delta(a) = 1$ then $\delta$ is a discrete distribution.
We only consider discrete (sub)distributions and we shall simply refer to them as (sub)distributions from now on.
Any (sub)distribution $\delta$ can be written as $\lmulti \delta(a) : a \rmulti_{a \in \supp(\delta)}$. The set of subdistributions over $A$ is denoted by $\mathcal{D}(A)$.
Note that $\mathcal{D}(A)$ is closed under convex combinations
$\sum_i p_i \cdot \delta_i \triangleq \lambda a. \sum_i p_i \delta_i(a)$
for countably many probabilities $p_i \in [0,1]$ such that $\sum_{i} p_i \leq 1$.
The notion of expectation of a function $f : A \to S$, where $S$ is a Kegelspitze, is defined for a given subdistribution $\delta$ over $A$ by
$\E{\delta}{f} \triangleq \Sigma_{a \in \supp(\delta)} \delta(a) \cdot f(a).$

A \emph{(weighted) Probabilistic Abstract Reduction System} (PARS) on $A$ is a ternary relation $\cdot \toop{\cdot} \cdot \subseteq A \times \Rpos \times \mathcal{D}(A)$. For $a \in A$, a rule $a \toop{c} \lmulti \delta(b) : b \rmulti_{b \in A}$ indicates that $a$ reduces to $b$ with probability $\delta(b)$ and cost $c \in \Rpos$. Given two objects $a$ and $b$,
$a \toop{c} \lmulti 1 : b \rmulti$ will be written $a \toop{c} b$ for brevity.
For simplicity, we consider only deterministic PARSs $\to$, i.e., $a \toop{c_1} \delta_1$ and $a \toop{c_2} \delta_2$ implies $c_1 = c_2$ and $\delta_1 = \delta_2$.
An object $a \in A$ is called terminal if there is no rule $a \toop{c} \delta$,
which we write as $a \not\to$.

Every deterministic PARS $\to$ over $A$ can be lifted to a ternary weighted reduction relation $\cdot \tomulti{\cdot} \cdot \ \subseteq \mathcal{D}(A) \times \Rpos \times \mathcal{D}(A)$ in a natural way,
see \Cref{fig:wrr}.
\begin{figure}[!h]
\begin{prooftree}
\hypo{a \not\to\phantom{\tomulti{c }}}
\infer1[(Term)]{\lmulti 1: a \rmulti \tomulti{0} \lmulti 1: a \rmulti}
\end{prooftree}
\quad
\begin{prooftree}
\hypo{a \toop{c} \delta}
\infer1[(Mono)]{ \lmulti 1: a\rmulti \tomulti{c } \delta  }
\end{prooftree}
\\[0.2cm]
\begin{prooftree}
\hypo{\delta_i \tomulti{c_i} \epsilon_i}
\hypo{\sum_i p_i \leq 1}
\infer2[(Muti)]{\sum_{i} p_i \cdot \delta_i \tomulti{\sum_i p_i c_i} \sum_{i} p_i \cdot \epsilon_i}
\end{prooftree}
\caption{Weigthed reduction relation induced by PARS.}
\label{fig:wrr}
\end{figure}
A reduction step $\delta \tomulti{c} \epsilon$ indicates that the subdistribution of objects $\delta$ evolves
to a subdistribution of reducts $\epsilon$
in one step, with an expected cost of $c$.
Note that since $\to$ is deterministic, so is the reduction relation $\tomulti{\cdot}$.
We denote by $\delta \tomulti{c}_n \epsilon$ the $n$-fold ($n \geq 0$) composition of $\tomulti{\cdot}$ with expected cost $c$,
defined by $\delta \tomulti{c}_n \epsilon$ if $\delta \tomulti{c_1} \cdots \tomulti{c_n} \epsilon$ and $c = \sum_{i=1}^n c_i$. In particular,  $\delta \tomulti{0}_0 \delta$.

Let us illustrate these notions on a small example.
\begin{example}
  We fix objects $A = \mathbb{Z} \cup \{\mathsf{geo}(n) \mid n \in \mathbb{Z}\}$.
  Consider the PARS $\to_{\mathsf{geo}}$ over $A$ defined through the rules
  \[
    \mathsf{geo}(n) \toop{1}_{\mathsf{geo}} \lmulti \sfrac{1}{2}: n+1, \sfrac{1}{2} : \mathsf{geo}(n + 1) \rmulti \qquad (n \in \mathbb{Z}),
  \]
  stating that $\mathsf{geo}(n)$ increments its argument and then either returns or recurs, in
  each case with probability $\sfrac{1}{2}$ and cost one.
  Starting from $\mathsf{geo}(0)$, this PARS admits precisely one infinite reduction sequence
  \begin{align*}
    \lmulti 1 : \mathsf{geo}(0) \rmulti
    & \tomulti{1}_{\mathsf{geo}} \lmulti \sfrac{1}{2} : 1, \sfrac{1}{2} : \mathsf{geo}(1) \rmulti \\
    & \tomulti{\sfrac{1}{2}}_{\mathsf{geo}} \lmulti \sfrac{1}{2} : 1, \sfrac{1}{4} : 2, \sfrac{1}{4} : \mathsf{geo}(2) \rmulti \\
    & \tomulti{\sfrac{1}{4}}_{\mathsf{geo}} \lmulti \sfrac{1}{2} : 1, \sfrac{1}{4} : 2, \sfrac{1}{8} : 3, \sfrac{1}{8} : \mathsf{geo}(3) \rmulti \\
    & \tomulti{\sfrac{1}{8}}_{\mathsf{geo}} \cdots
  \end{align*}

  This sequence approaches the distribution $\lmulti \sfrac{1}{2^{n}} : n \rmulti_{n>0}$
  of terminal objects in $\mathbb{Z}$,
  with an expected cost of $\sum_{i = 0}^\infty \sfrac{1}{2^i} = 2$.
\end{example}

As indicated in this example,
for every $\delta \in \mathcal{D}(A)$ there is precisely one infinite sequence
$\delta = \delta_0 \tomulti{c_0} \delta_1 \tomulti{c_1} \delta_2  \tomulti{c_2} \cdots$
gradually approaching a \emph{normal form distribution} of terminal objects with
an expected cost of $\sum_{i=0}^\infty c_i$.\footnote{
  This infinite sum is always defined, since costs $c_i$ are non-negative.
}
Note that this
normal form distribution can be a proper subdistribution --- in which case the PARS
is not almost-surely terminating --- and that the cost can be infinite.

Based on these intuitions, for an object $a \in A$,
we define the expected cost function
$\ecost_\to : A \to \Rext$ by
\[
  \ecost_\to(a) \triangleq \sup_{n \in \mathbb{N}}\{ c \mid \lmulti 1: a\rmulti \tomulti{c}_n \delta \},
\]
and the normal form function $\nf_\to : A \to \mathcal{D}(A)$ by
\[
  \nf_\to(a) \triangleq \sup_{n \in \N}\{ \delta \restrict_{term} \mid \lmulti 1: a\rmulti \tomulti{c}_n \delta \},
\]
where $\delta \restrict_{term}$ is the restriction of $\delta$ to terminal objects, i.e., $\delta \restrict_{term} \triangleq \{\delta(a) : a \ | \ a \not\to \}_{a \in \supp(\delta)}$, and the supremum is taken w.r.t. the pointwise order of subdistributions. Note that $\nf_\to(a)$ is well-defined,
which essentially follows from the fact that $(\delta_n \restrict_{term})_{n \in \N}$, for $\delta_n$ such that
for $\lmulti 1: a\rmulti \tomulti{c}_n \delta_n$, is a monotonically increasing sequence, by definition of $\tomulti{c}$.

\paragraph*{Quantum Programs as PARSs}

\begin{figure*}[h]
\hrulefill

$
\begin{prooftree}
\hypo{\phantom{\toop{r}}}
\infer1[(Skip)]{\ltuple\skp,s,\ket{\varphi}\rtuple \toop{0} \ltuple s,\ket{\varphi}\rtuple}
\end{prooftree}
\quad
\quad
\begin{prooftree}
\hypo{\phantom{\toop{r}}}
\infer1[(Exp)]{\ltuple \x \asg \e,s,\ket{\varphi}\rtuple \toop{0} \ltuple s[\x := \sem{\e}^s], \ket{\varphi}\rtuple }
\end{prooftree}
\quad
\quad
\begin{prooftree}
  \hypo{\phantom{\ltuple\cmd_1 ,\ket{\varphi}\rtuple \toop{r} \ltuple \ket{\varphi}\rtuple}}
\infer1[(Op)]{\ltuple\overline{\q} \oasg \ope,s,\ket{\varphi}\rtuple \toop{0} \ltuple s, \ope_{\overline{\q}}\ket{\varphi}\rtuple}
\end{prooftree}
$
\\
$
\begin{prooftree}
\hypo{\phantom{\toop{r}}}
\infer1[(Meas)]{\ltuple \x \asg \meas(\q),s,\ket{\varphi}\rtuple \toop{0} \lmulti p_k^\q\ket{\varphi}: \ltuple s[\x := k], \oper{M}_k^\q\ket{\varphi}\rtuple \rmulti_{k \in \{0,1\}}}
\end{prooftree}
$
\quad
$
\quad \quad
\begin{prooftree}
\hypo{\phantom{\ltuple\cmd_1 ,\ket{\varphi}\rtuple \toop{r} \ltuple\epsilon,\ket{\varphi}\rtuple}}
\infer1[(Cons)]{\ltuple\cons(\nexp),s,\ket{\varphi}\rtuple \toop{\max(\sem{\nexp}^s,0)} \ltuple s,\ket{\varphi}\rtuple}
\end{prooftree}
$
\\[0.2cm]
$
\begin{prooftree}
\hypo{\ltuple\cmd_1 ,s,\ket{\varphi}\rtuple \toop{c} \lmulti p_i : \ltuple\cmd_1^i,s^i,\ket{\varphi^i}\rtuple \rmulti_{i \in I} \cup \lmulti q_j : \ltuple s^j,\ket{\varphi^j}\rtuple \rmulti_{j \in J}}
\infer1[(Seq)]{\ltuple\cmd_1 \sep \cmd_2,s,\ket{\varphi}\rtuple \toop{c} \lmulti p_i: \ltuple\cmd_1^i  \sep \cmd_2,s^i,\ket{\varphi^i}\rtuple \rmulti_{i \in I} \cup \lmulti q_j: \ltuple \cmd_2,s^j,\ket{\varphi^j}\rtuple \rmulti_{j \in J}}
\end{prooftree}
$
\\[0.2cm]
$
\begin{prooftree}
\hypo{\sem{\bexp}^s \in \{0,1\} \phantom{\ket{\varphi^i}}}
\infer1[(Cond)]{\ltuple\ifa (\bexp) \{\cmd_1\} \elsea \{\cmd_0\},s,\ket{\varphi}\rtuple \toop{0} \ltuple\cmd_{\sem{\bexp}^s},s,\ket{\varphi}\rtuple}
\end{prooftree}
$
\\[0.2cm]
$
\begin{prooftree}
\hypo{\sem{\bexp}^s =0}
\infer1[(Wh$_0$)]{\ltuple\while(\bexp)\{ \cmd \} ,s,\ket{\varphi}\rtuple \toop{0} \ltuple  s, \ket{\varphi}\rtuple}
\end{prooftree}
\quad
\quad
\begin{prooftree}
\hypo{\sem{\bexp}^s =1}
\infer1[(Wh$_1$)]{\ltuple\while(\bexp)\{ \cmd \} ,s,\ket{\varphi}\rtuple \toop{0} \ltuple\cmd \sep \while(\bexp)\{ \cmd \}, s,\ket{\varphi}\rtuple}
\end{prooftree}
$
\\[10pt]
\hrulefill
\caption{Operational semantics in terms of PARS.}
\label{fig:os}
\end{figure*}
We now endow quantum programs with an operational semantics defined through a PARS,
operating on pairs of classical and quantum states.

Let $\mathbb{C}$ denote the set of complex numbers. 
Given a set $Q$ of $n$ qubit variables, let $\mathcal{H}_{Q}$ be the Hilbert space $\mathbb{C}^{2^{n}}$ of $n$ qubits\footnote{We assume $Q$ to be a totally ordered set so that the smallest element of $Q$ corresponds to the first qubit of $\mathcal{H}_{Q}$ and so on.}.
We use Dirac notation, $\ket{\varphi}$, to denote a quantum state of $\mathcal{H}_{Q}$. Any state $\ket{\varphi}$ can be written as $\Sigma_{b \in \{0,1\}^{n}} \alpha_b \ket b$, with $\alpha_b \in \mathbb{C}$, and $\Sigma_{b \in \{0,1\}^{n}} \size{\alpha_b}^2 = 1$. $\bra{\varphi}$ is the conjugate transpose of $\ket{\varphi}$, i.e., $\bra{\varphi} \triangleq \ket{\varphi}^\dagger$. $\braket{\varphi}{\psi}\triangleq {\bra{\varphi}} \ket{\psi}$ and $\ketbra{\varphi}{\psi}$  denote the inner product and outer product of $\ket{\varphi}$ and $\ket{\psi}$, respectively.
The norm of a vector is defined by $\norm{\ket{\varphi}} \triangleq \sqrt{\braket{\varphi}{\varphi}}$. We define (linear) operators over $\mathcal{H}_{Q}$ as linear maps. Hence an operator will be represented by a square matrix whose dimension is equal to the dimension of $\mathcal{H}_{Q}$. Given $m \geq 1$, let $I_{m}$ be the  $m \times m$ identity matrix and $\otimes$ be the standard Kronecker product on matrices.

Assume that $Q=\{\q_1,\ldots,\q_n\}$. For $k \in \{0,1\}$, let $\ket{k}_{\q_i} \in \mathcal{H}_{Q}$ be defined by $\ket{k}_{\q_i} \triangleq I_{2^{i-1}} \otimes \ket{k} \otimes I_{2^{n-i}}$ and let $\bra{k}_{\q_i}$ be its conjugate transpose. The measurement of a qubit $\q\in Q$ of a state $\ket \varphi \in \mathcal{H}_{Q}$ produces the classical outcome $k\in \{0,1\}$ with probability $p^\q_k \ket{\varphi}$, and transforms the quantum state $\ket \varphi$ into $ \oper{M}^\q_k\ket{\varphi}$, where $\oper{M}^\q_k : \mathcal{H}_{Q} \to \mathcal{H}_{Q}$ is defined as
$$\oper{M}^\q_k \triangleq  \ket{\varphi} \mapsto \frac{\ket{k}_\q \bra{k}_\q \ket{\varphi}}{\norm{\bra{k}_\q \ket{\varphi}}} $$
and $p^\q_k : {\mathcal{H}_{Q}} \to [0,1]$ is defined as $p^\q_k \triangleq \ket{\varphi} \mapsto  \norm{\bra{k}_\q \ket{\varphi}}^2$. 

The classical state is modelled as a (well-typed) \emph{store} $s$.
For two given sets $B$ and $V$ of Boolean and numerical variables, a (classical) \emph{store} $s$ is a pair of maps $(s^B,s^V)$ such that $s^B : B \to \{0,1\}$ and $s^V : V \to \mathbb{Z}$. The domain of $s$, noted $dom(s)$, is defined by $dom(s) \triangleq B \cup V$. Given a store $s=(s^B,s^V)$, we let $s[\x^\Var := k]$ (resp. $s[\x^\Bool := k]$, $k \in \{0,1\}$) be the store obtained from $s$ by updating the value assigned to $\x$ in the map $s^V$ (resp. $s^B$) to $k$. Define also $s(\x^\Var)\triangleq s^V(\x^\Var)$ and $s(\x^\Bool)\triangleq s^B(\x^\Bool)$. Given a store $s$, let $\sem{-}^s$ be the map associating to each expression $\e$ (and such that $\Bool(\e) \cup \Var(\e) \subseteq dom(s)$) of type $\Var$, a value in $\mathbb{Z}$, and to each expression $\e$ of type $\Bool$ a value in $\{0,1\}$, and defined in a standard way. For example $\sem{\x}^s \triangleq s(\x)$, $\sem{n}^s \triangleq n$, $\sem{\true}^s \triangleq 1$, etc.

A \emph{state} $\sigma$ is a pair $\ltuple s,\ket{\varphi} \rtuple$ consisting of a store $s$ and a quantum state $\ket{\varphi}$.
A \emph{configuration} $\mu$ for statement $\cmd$ has the form $\ltuple \cmd, \sigma \rtuple$, sometimes written as
$\ltuple \cmd, s,\ket{\varphi} \rtuple$ for $\sigma=\ltuple s,\ket{\varphi} \rtuple$.
Let $\memst$ and $\conf$ be the set of states and the set of configurations, respectively.
A configuration $\ltuple\cmd, s,\ket{\varphi} \rtuple$ is well-formed with respect to the sets
of variables $B,V,Q$ if $\Bool(\cmd) \subseteq B$, $\Var(\cmd) \subseteq V$,  $\Qubits(\cmd) \subseteq Q$, $dom(s) = B \cup V$, and $\ket{\varphi} \in \mathcal{H}_Q$.
%
Throughout the paper, we only consider configurations that are well-formed with respect to the sets of variables of the program under consideration.

The operational semantics is described in Figure~\ref{fig:os} as a PARS $\to$ over objects in $\conf \cup \memst$,
where precisely the objects in $\memst$ are terminal.
Rule (Cons) evaluates the arithmetic expression provided as argument to a cost, an integer, and annotates the reduction with this cost, whenever it is a positive integer (otherwise the cost is $0$). The state of a configuration can only be updated by the three rules (Exp), (Op), and (Meas). Rule (Exp) updates the classical store in a standard way. Rule (Op) updates the quantum state to a new quantum state $\ope_{\overline{\q}}\ket{\varphi}$, where $\ope_{\overline{\q}}$ is the map that applies the unitary operator $\ope$ to qubits in $\overline{\q}=\q_1,\ldots,\q_{ar(\ope)}$ and tensoring the map with the identity on all other qubits to match the dimension of $\ket{\varphi}$. Rule (Meas) performs a measurement on qubit $\q$. This rule returns a distribution of configurations corresponding to the two possible outcomes, $k=0$ and $k=1$, with their respective probabilities $p_k^\q\ket{\varphi}$ and, in each case, updates the classical store and the quantum state accordingly. Rule~(Seq) governs the execution of a sequence of statements $\cmd_1\sep\cmd_2$.
The rule accounts for potential probabilistic behavior when $\cmd_1$ performs a measurement and it is otherwise standard.
All the other rules are standard.

For a statement $\cmd$, we overload the notion of expected cost function
and define $\ecost_\cmd : (\Rext)^\memst$ by
\[
  \ecost_\cmd( \sigma  ) \triangleq \ecost_\to (\cmd, \sigma).
\]
Moreover, the function
$\evalue_\cmd: \CSd^\memst \to \CSd^\memst$ defined
by
\[
  \evalue_\cmd(f) (\sigma)
  \triangleq \E{\nf_\to ( \cmd, \sigma)}{f}
\]
gives the expected value of $f$ on the subdistribution of terminal states
obtained by executing $\cmd$ on state $\sigma$. Note that this function is well-defined,
as $\nf_\to(\cmd,\sigma)$ is a sub-distribution over $\memst$.

\begin{example}\label{ex:cointossos}
  Consider the program from \Cref{ex:cointoss}.
  Let $\cmd$ refer to the while loop.
  %
  On a state $\ltuple s, \ket{\varphi} \rtuple$ such that $\ket{\varphi}= \alpha \ket{0}+\beta\ket{1}$ (with $\size{\alpha}^2 + \size{\beta}^2 = 1$) and $s(\x)=1$, it holds that:

\begin{align}
  &\mparbox{2mm}{ \delta_0 \triangleq \lmulti 1 : \ltuple\cmd ,s,\ket{\varphi}\rtuple \rmulti}\nonumber \\
   &\tomulti{0}\small \lmulti 1 : \ltuple \q \oasg \oper{H} \sep \x \asg \meas(\q) \sep \cons(1) \sep \cmd, s,\ket{\varphi}\rtuple \rmulti \label{e1}\\
 &\tomulti{0} \lmulti 1 :\ltuple \x \asg \meas(\q) \sep \cons(1) \sep \cmd,s, \oper{H}\ket{\varphi}\rtuple \rmulti \label{e2}\\
 &\tomulti{0}\! \lmulti p_k\! : \! \ltuple \cons(1) \sep \cmd,s[\x := k] , \ket{k} \rtuple \rmulti_{k \in \{0,1\}} \label{e3} ,
\end{align}
with $p_0= \frac{\size{\alpha+\beta}^2}{2}$, and $p_1 = \frac{\size{\alpha-\beta}^2}{2}$. The above reductions are obtained by applying rules of ~\Cref{fig:os} together with rule (Mono) of ~\Cref{fig:wrr}: (Wh$_1$) for (\ref{e1}); (Op) and (Seq) for (\ref{e2}); (Meas) and (Seq) for (\ref{e3}).

Moreover, by rules (Cons), (Seq) , and  (Mono), $\forall k \in \{0,1\}$,
 {\small
   \begin{align}
     \lmulti 1 :\ltuple \cons(1) \sep \cmd,s[\x := k] , \ket{k} \rtuple \!\rmulti \tomulti{1} \!\lmulti 1 :\ltuple \cmd,s[\x := k] , \ket{k} \rtuple \!\rmulti . \label{e4}
\end{align}
 }
 Consequently, using (\ref{e1})-(\ref{e4}) and rule (Multi) of ~\Cref{fig:wrr}:
  {
 \[
 \lmulti 1 : \ltuple\cmd ,s,\ket{\varphi}\rtuple \rmulti\tomulti{1}_4\lmulti p_k : \ltuple  \cmd,s[\x := k] , \ket{k} \rtuple \rmulti_{k \in \{0,1\}}, 
\]
}
as $p_0+p_1=1$. Iterating the above reduction, it holds that
\[
\lmulti 1 : \ltuple\cmd ,s,\ket{1}\rtuple  \rmulti \tomulti{1}_4 \lmulti \sfrac{1}{2} : \ltuple  \cmd,s[\x := k] , \ket{k} \rtuple \rmulti_{k \in \{0,1\}} \]
Moreover, as
\[
\lmulti 1: \ltuple\cmd ,s[\x := 0],\ket{0}\rtuple \rmulti  \tomulti{0} \lmulti 1: \ltuple s[\x := 0], \ket{0}\rtuple \rmulti,
\]
it holds that
\begin{align*}
\delta_0 & \tomulti{1}_4 \lmulti p_0 :  \ltuple  \cmd,s[\x := 0] , \ket{0} \rtuple , p_1 :  \ltuple  \cmd,s, \ket{1} \rtuple \rmulti\\
&\tomulti{p_1}_4  \lmulti p_0 + \sfrac{p_1}{2}: \ltuple s[\x := 0], \ket{0}\rtuple ,  \sfrac{p_1}{2}: \ltuple  \cmd,s , \ket{1} \rtuple \rmulti \\
&\tomulti{\frac{p_1}{2}}_{4}  \lmulti p_0 + \sfrac{p_1}{2}+ \sfrac{p_1}{4}: \ltuple s[\x := 0], \ket{0}\rtuple ,  \sfrac{p_1}{4}: \ltuple  \cmd,s , \ket{1} \rtuple \rmulti\\
&\tomulti{\frac{p_1}{4}}_{4} \ldots
\end{align*}
The expected cost and the normal form are obtained as follows, by reasoning about the asymptotic behaviour.
\[
  \ecost_{CT(\q)} \ltuple s,{\small\begin{pmatrix}
\alpha\\
\beta
\end{pmatrix}}\,\rtuple= \sup_{n \in \N}\left\{ 1+p_1\textstyle\sum_{i = 0}^n \frac{1}{2^i}\right\}=1+\size{\alpha-\beta}^2
\]
\begin{align*}
  \nf_{CT(\q)}\ltuple s,\ket{\varphi}\rtuple  &= \sup_{n \in \N}\left\{ \lmulti p_0+p_1 \textstyle\sum_{i = 1}^n \frac{1}{2^i}: \ltuple s[\x := 0], \ket{0}\rtuple\rmulti \right\}\\
  &= \lmulti 1: \ltuple s[\x := 0], \ket{0}\rtuple\rmulti
\end{align*}
Hence, $\evalue_{CT(\q)}(f) (s,\ket{\varphi})$, the expected value of $f$ after executing $CT(\q)$ on $(s,\ket{\varphi})$, is equal to
$f(s[\x := 0], \ket{0})$.
\end{example}

\section{Quantum Expectation Transformers}
\label{sec:qet}

We now revisit the expectation transformer approach for the quantum programming language introduced in ~\Cref{s:l}.
$\emph{Expectations}$ will be functions from the set of (classical and quantum) memory states to cost structures, i.e., functions in $\CSd^\memst$, for a given cost structure $\CSd$. The \emph{quantum expectation transformer} $\qet{\cdot}{\cdot}$ is then defined in terms of a program semantics mapping expectations to expectations in a continuation passing style. Specializing the cost structure yields several quantum expectation transformers such as the \emph{quantum expected value transformer} $\qev{\cdot}{\cdot}$ and the \emph{quantum expected cost transformer} $\qect{\cdot}{\cdot}$. After exhibiting several laws and properties of these transformers, we show their soundness and their adequacy.


\subsection{Definition}
Before defining expectation transformers, we introduce some preliminary notations in order to lighten the presentation.

\begin{notations}

For any expression $\e$, $\sem{\e}$ is a shorthand notation for the function $\lambda \ltuple s, \ket{\varphi} \rtuple.\sem{\e}^s \in (\Rext)^\memst$ and, for any $c \in \Rext$, $\underline{c}$ is the function in $(\Rext)^\memst$ defined by $\underline{c} \triangleq \lambda \sigma.c$.
To avoid notational overhead, 
we frequently use point-wise extensions of operations on $\CSp$ and $\Rext$ to functions. E.g., for $p \in [0,1]^\memst$, $f,g \in \CSd^\memst$,
  $f \up{p} g$ denotes the function $\lambda \sigma. f(\sigma) +_{p(\sigma)} g(\sigma)$.

We will also use $f[\x := \e]$ for the expectation mapping $\ltuple s, \ket{\varphi} \rtuple$ to $f\ltuple s[\x := \sem{\e}^s], \ket{\varphi} \rtuple$,
and similarly, for a given function $\oper{M} : \mathcal{H}_{Q} \to \mathcal{H}_{Q}$,  $f[\oper{M}]$ maps $\ltuple s, \ket{\varphi} \rtuple$ to $f\ltuple s, \oper{M}\ket{\varphi} \rtuple$. Finally, $f[\x := \e;\oper{M}]$ stands for $(f[\x := \e])[\oper{M}]$.
\end{notations}

\begin{figure*}[!h]
\hrulefill
\begin{align*}
  \qet{\epsilon}{f} &\triangleq f &  \qet{\x \asg \meas(\q)}{f} &\triangleq f[\x := 0;\oper{M}_0^\q] \up{p_0^\q} f[\x := 1;\oper{M}_1^\q]\\
  \qet{\skp}{f} & \triangleq f & \qet{\cons(\nexp)}{f} &\triangleq \umax(\sem{\nexp},\underline{0}) \uCSp f \\
  \qet{\x \asg \e}{f} &\triangleq f[\x := \e]   & \qet{\cmd_1 \sep \cmd_2}{f} &\triangleq \qet{\cmd_1}{\qet{\cmd_2}{f}} \\
\qet{\overline{\q} \oasg \ope }{f} &\triangleq f[\ope_{\overline{\q}}]  & \qet{\ifa (\bexp) \{\cmd_1\} \elsea \{\cmd_2\} }{f} &\triangleq \qet{\cmd_1}{f}\up{\sem{\bexp}} \qet{\cmd_2}{f}\\
 &  & \qet{\while(\bexp)\{ \cmd \}}{f} &\triangleq \lfp\left(\lambda F.\qet{\cmd}{F} \up{\sem{\bexp}} f \right)
\end{align*}
\hrulefill
\caption{Quantum Expectation Transformer $\qet{\cdot}{\cdot} : \Progs \to \CSd^\memst \to \CSd^\memst$.}
\label{fig:qet}
\end{figure*}

\begin{definition}
\label{def:qet}
  Let $(\CSd,\CSp)$ be a cost structure. The quantum expectation transformer
  \[
    \qet{\cdot}{\cdot} : \Progs \to \CSd^\memst \to \CSd^\memst
  \]
  is defined inductively in Figure~\ref{fig:qet}.
\end{definition}

\begin{definition}[Quantum expectation transformers instances]\label{def:qet-instances}
  \begin{enumerateenv}
  \item Taking the cost structure $([0,1],+_{\mathtt{f}})$ yields a weakest precondition transformer
    \[
      \qwp{\cdot}{\cdot} : \Progs \to [0,1]^\memst \to [0,1]^\memst ,
    \]
    for probabilistic pre-condition reasoning.
  \item Taking the cost structure $(\CSd,+_{\mathtt{f}})$, for any Kegelspitze $\CSd$,
    yields an expected value transformer
    \[
      \qev{\cdot}{\cdot} : \Progs \to \CSd^\memst \to \CSd^\memst .
    \]
  \item Taking the cost structure $(\Rext,+)$ yields an expected cost transformer
    \[
      \qect{\cdot}{\cdot} : \Progs \to (\Rext)^\memst \to (\Rext)^\memst
     .
    \]
  \end{enumerateenv}
\end{definition}

\subsection{Properties}
\label{sub:prop}
\begin{figure*}[t]
\hrulefill
\begin{align*}
  \quad
 & \law[idents:cont]{continuity}                &  & \textstyle \qet{\cmd}{\sup_i f_i} = \sup_i \qet{\cmd}{f_i} \text{ for any $\omega$-chain $(f_i)_i$} \\
 & \law[idents:mono]{monotonicity}              &  & f \leq g \implies \qet{\cmd}{f} \leq \qet{\cmd}{g}                                                                         \\
 & \law[idents:distri]{distributivity}          &  & p \in [0,1] \implies \qet{\cmd}{f \up{p} g} = \qet{\cmd}{f}\up{p} \qet{\cmd}{g}                                              \\
 & \law[idents:ui]{upper invariant}             &  & (\sem{\neg \bexp} \ucdot f \leq g\ \wedge\ \sem{\bexp} \ucdot \qet{\cmd}{g} \leq g) \implies \qet{\while(\bexp)\{\cmd\}}{f} \leq g 
\end{align*}
\hrulefill
\caption{Universal laws derivable for the quantum expectation transformer.}
\label{fig:idents}
\end{figure*}

The quantum expectation transformer satisfies several useful laws (\Cref{fig:idents}) and these laws are comparable to those found in \cite{KKMO16}.

\begin{theorem}\label{thm:idents}
  All universal laws listed in \Cref{fig:idents} hold.
\end{theorem}

 The \ref{idents:mono} Law permits us to reason modulo upper-bounds: actual costs can be always substituted by upper-bounds.
 It is in fact an immediate consequence from the continuity law, itself essential for the well-definedness of the transformer on while loops.
 The \ref{idents:distri} Law is a direct consequence of the laws on cost structures.
 The \ref{idents:ui} Law generalises the corresponding law by \cite{KKMO16},
 itself a generalisation of the notion of invariant stemming from Hoare calculus from predicates to cost functions.
 It constitutes a complete proof rule for finding closed form upper-bounds for loops,
 and based on the observation that any prefix-point
 --- as given with $g$ in the pre-condition --- is an upper bound to the least-prefixed point
 of a functional --- in our case the expected cost of the loop (w.r.t. $f$).

\begin{example}\label{ex:cointossqect}
Let us search for a cost expectation of the program of Example~\ref{ex:cointoss}.
Recall that $\cmd_0$ is the body of the while loop statement and that the considered cost structure is $(\Rext,+)$. By \ref{idents:ui} Law (see \Cref{fig:idents}), it suffices to find an expectation $g$ satisfying the following inequalities
\begin{align}
  \sem{\neg \x} \cdot \underline{0} &\leq g \label{ee1}\\
  \sem{\x} \cdot \qect{\cmd_0}{g} & \leq g \label{ee2}
\end{align}
in order to compute an upper bound on the expectation $\qect{\cmd}{\underline{0}}$ of the while loop statement $\cmd$.

Using rules of Figure~\ref{fig:qet}, the following equalities hold,
as can be verified directly:
\begin{align*}
g_1 &\triangleq \qect{\cons(1)}{g} = \underline{1}+g \\
g_2 &\triangleq \qect{\x \asg \meas(\q) }{g_1} \\
    & = g_1[\x := 0; \oper{M}_0^\q]   \up{p_0^\q} g_1[\x := 1;\oper{M}_1^\q] \\
    & =\lambda \ltuple s,\ket{\varphi}\rtuple.\textstyle\sum_{k \in \{0,1\}} p_k^\q (\ket{\varphi}) g_1\ltuple s[ \x := k],\ket{k} \rtuple )\\
g_3 &\triangleq \qect{\q \oasg \oper{H}}{g_2} =g_2[\oper{H}] \\
    & =\lambda \ltuple s,\ket{\varphi}\rtuple.\textstyle\sum_{k \in \{0,1\}} p_k^\q (\oper{H}\ket{\varphi}) (\underline{1}+g)\ltuple s[ \x := k],\ket{k} \rtuple .
\end{align*}
Now, we set
\[
  g \ltuple s, \begin{pmatrix}
\alpha\\
\beta
\end{pmatrix} \rtuple \triangleq \sem{\x} \cdot (1+ \size{\alpha-\beta}^2).
\]
It holds that $\oper{H} \begin{pmatrix}
\alpha\\
\beta
\end{pmatrix} = \frac{\alpha+\beta}{\sqrt{2}} \ket{0} + \frac{\alpha-\beta}{\sqrt{2}}\ket{1}$ and that:
\begin{align*}
\mparbox{3mm}{\qect{\cmd_0}{g}\ltuple s,\begin{pmatrix}
\alpha\\
\beta
\end{pmatrix}\rtuple}\\
&  =
g_3 \ltuple s,\begin{pmatrix}
\alpha\\
\beta
\end{pmatrix}\rtuple \\
&= {\frac{\size{\alpha+\beta}^2}{2}}(1+0)  + \frac{\size{\alpha-\beta}^2}{2}(1+1+(0-1)^2)\\
& = 1+\size{\alpha-\beta}^2 = g\ltuple s,\begin{pmatrix}
\alpha\\
\beta
\end{pmatrix}\rtuple.
\end{align*}
Therefore, (in)equalities (\ref{ee1}) and (\ref{ee2}) are satisfied by $g$. It follows that
\begin{align*}
\qect{CT(\q)}{\underline{0}} & \leq \qect{\x^\Bool \asg \true }{g} \\
&=g [\x := 1] = \lambda \ltuple s, \begin{pmatrix}
\alpha\\
\beta
\end{pmatrix}\rtuple.1+ \size{\alpha-\beta}^2
\end{align*} is the cost expectation of the program.
Note that, in this case, this bound is exact.
\end{example}

\subsection{Soundness and Adequacy}
\label{sub:adequacy}

We first give a semantic counterpart
to the quantum expectation transformer.
To this end, for $f \in \CSd^\memst$ let
$\QETT{f} \in \CSd^{\conf \cup \memst}$
be the least function (in the pointwise order inherited from $\CSd$), such that:
\begin{align*}
  \QETT{f}(\mu) & =
               \begin{cases}
                 f(\mu) & \text{if $\mu \in \memst$}\\
                 c \CSp \E{\delta}{\QETT{f}} & \text{if $\mu \in \conf$ and $\mu \toop{c} \delta$}
               \end{cases}
\end{align*}
(see Appendix \ref{app:QET}). Finally, we overload notation and set
\[
  \QET{\cmd}{f}(\sigma) \triangleq  \QETT{f} \ltuple \cmd,\sigma\rtuple .
\]
The following correspondence is not difficult to establish.
\begin{lemma*}{l:QET}
  For all $\cmd \in \Cmds$ and $f \in \CSd^\memst$,\\
  $ \QET{\cmd}{f}
    = \ecost_\cmd \CSp \evalue_\cmd(f)$.
\end{lemma*}
We now show that the quantum expectation transformer coincides with its semantic counterpart.
Via this correspondence, the above lemma allows us to relate the quantum expectation transformer
--- and its derivates from \Cref{def:qet-instances} --- to the cost and the semantics of the considered programs.
To establish the link, we make use of the following two identities.
\begin{lemma*}{l:QET:main-props}
  The following identities hold.
  \begin{enumerateenv}
  \item\label{l:QET:main-props:seq} $\QET{\cmd_1 \sep \cmd_2}{f} = \QET{\cmd_1}{\QET{\cmd_2}{f}}$; and
  \item\label{l:QET:main-props:while} $\QET{\while(\bexp)\{ \cmd \}}{f} = \lfp(\lambda F. \QET{\cmd}{F} +_{\sem{\bexp}} f)$.
  \end{enumerateenv}
\end{lemma*}

\begin{theorem}[Soundness]\label{t:soundness}
  For all $\cmd \in \Cmds$, $\sigma \in \memst$ and $f \in \CSd^\memst$, $\qet{\cmd}{f}(\sigma) = \QET{\cmd}{f}(\sigma)$.
\end{theorem}
\begin{proof}
  The theorem is proven by induction on $\cmd$. Almost all cases follow by definition.
  The only two non-trivial cases, those of command composition and loops,
  follow from the induction hypothesis by using \Cref{l:QET:main-props}.
\end{proof}

This theorem and \Cref{l:QET} immediately show how to recover the expected cost and expected value of programs.
\begin{corollary}[Adequacy]\label{c:ect-soundness}
  The following identities hold, for all $\cmd \in \Cmds$, $ \sigma \in \memst$ and $f \in  \CSd^\memst$.
  \begin{enumerate}
  \item $\qect{\cmd}{\underline{0}}  (\sigma) = \ecost_\cmd (\sigma)$; and
  \item $\qev{\cmd}{f}  (\sigma) = \evalue_\cmd(f) (\sigma)$.
  \end{enumerate}
\end{corollary}

\begin{example}
  We illustrate the soundness theorem on our simple leading example (\Cref{ex:cointossos}). As calculated in \Cref{ex:cointossos} and \Cref{ex:cointossqect}, we have
  \begin{equation*}
    \ecost_{CT(\q)} = \lambda \ltuple s, \begin{pmatrix}
      \alpha\\
      \beta
    \end{pmatrix}\rtuple. 1+\size{\alpha-\beta}^2 =
    \qect{CT(\q)}{\underline{0}}.
  \end{equation*}
  %
\end{example}


\subsection{Relationship to Denotational Semantics}
\label{sub:denotational-semantics}

As a special case of our quantum expectation transformer, we can define a
quantum denotational semantics for our language. Recall that the formation
conditions for configurations are defined with respect to three sets $B,V,Q$ of
Boolean, numerical and quantum variables, respectively. Let $B=\{b_1, \ldots, b_n\}$, $V = \{v_1, \ldots v_m\}$ and $Q = \{q_1, \ldots q_k\}.$ We can define a Kegelspitze $\CSk$ which serves as the semantic domain for our programs by setting
$\CSk \eqdef \{0,1\}^n \times \mathbb Z^m \to D_{2^k},$
where the order and convex structure of $\CSk$ is inherited pointwise from $D_{2^k}$ (see Example \ref{ex:density-matrix}). A cost structure $\mathcal K$ is now obtained by equipping $\CSk$ with
the forgetful cost addition $\mathcal K = (\CSk, +_{\mathtt f}).$ With this choice of cost structure, our quantum expectation transformer $\mathtt{qet}$ from Definition \ref{def:qet} yields a \emph{quantum denotational semantics} transformer
\[ \qdst{\cdot}{\cdot} : \Progs \to \CSk^\memst \to \CSk^\memst . \]
Recall that a quantum denotational semantics consists in giving a mathematical interpretation of program configurations which is invariant under the operational semantics (in a probabilistic sense). This can be obtained from $\mathtt{qev}_\CSk$ by making a suitable choice for the
continuation. In particular, if we choose
\begin{align*}
  & h: \memst \to \CSk\\
  & h(s^B, s^V, \ket \varphi) = \lambda ((t_1, \ldots, t_n), (u_1, \ldots, u_m)). \\
  & \qquad
  \begin{cases}
    \ketbra{\varphi} & \text{if } t_i = s^B(b_i) \text{ for }  1 \leq i \leq n \text{ and } \\
    & u_j = s^V(v_j) \text{ for } 1 \leq j \leq m \\
    \mathbf 0 & \text{otherwise}
  \end{cases}
\end{align*}
then, a quantum denotational semantics
\[
  \flrb{-} : \conf \cup \memst \to \CSk
\]
can be defined by
$ \flrb{(\cmd, \sigma)} \eqdef \qdst{\cmd}{h}(\sigma) $ for configurations and $\flrb{\sigma} \eqdef h(\sigma)$ for program states (which are our notion of terminal objects).
Then, by \Cref{c:ect-soundness},
for any well-formed configuration $\mu = (\cmd, \sigma)$ we have that
\begin{align*}
  \flrb{\mu} = \qdst{\cmd}{h}  (\sigma)
  = \evalue_\cmd(h) (\sigma)
  = \mathbb E_{\nf_\to(\mu)} ( \flrb{-} ) .
\end{align*}
This shows that the denotational interpretation $\flrb{\mu}$ is equal to the (countable) convex sum of the interpretations of final states (i.e., terminal objects) that $\mu$ can reduce to. In this equation, each probability weight associated to a final state $\tau$ is given by the reduction probability of $\mu$ to $\tau$ as determined
by the operational semantics. This is precisely the statement of \emph{strong adequacy} in the denotational semantics of probabilistic \cite{strong-adequacy,commutative-monads} and quantum programming languages \cite{fossacs20,popl22}.

\section{Illustrating Examples}\label{s:ex}

In this section we present more intricate examples illustrating how cost analysis can be performed for quantum algorithms.
The analysis has a focuss on expected costs. Hence, the cost structure is fixed to be $(\Rext, +)$. Consequently, cost expectations will be functions in the set $(\Rext)^\memst$.

\begin{notations}
Throughout the following, we denote by $\kappa$ a \emph{classical expectation},
i.e., an expectation which satisfies $\kappa \ltuple s, \ket{\varphi} \rtuple = \kappa \ltuple s, \ket{\psi} \rtuple$ for all $s, \ket{\varphi} $, and $\ket{\psi}$. A classical expectation $\kappa$ thus only depends on the classical state.

We also define the following syntactic sugar:
\begin{itemize}
\item $\q \asg  \ket{0}$ for $ \x^\Bool \asg \meas(\q) \sep \ifa (\x) \{\ \q \oasg \oper{X} \} \elsea \{\skp\}$, where $\oper{X}$ is the unitary operator for negation (Pauli-$X$ gate), defined by $\oper{X} \ket{0} = \ket{1}$ and $\oper{X} \ket{1} = \ket{0}$.
\item $\q \asg  \ket{+}$ for $\q \asg  \ket{0} \sep \q \oasg \oper{H}$.
\end{itemize}
\end{notations}

\subsection{Repeat until success}\label{ex:1}
In this example, we consider a program that implements a Repeat-until-success algorithm and we show that our analysis can be used to infer upper bounds on the expected \emph{T-count}, i.e., the expected number of times the so-called T gate is used. 

\begin{figure}
  \begin{lstlisting}[caption={Repeat until success.}, label={ex2}, style=qwhile]
$RUS(\q')$ $\triangleq$ $\x^\Bool \asg \true \sep$
         $\while(\x) \{ $
           $\q \asg \ket{+}\sep$ $\tikzrom{mark1}$ $\quad \quad \quad $ $\tikzrom{mar1}$
           $\q \oasg \oper{T} \sep$
           $\cons(1) \sep$
           $\q,\q' \oasg \oper{CNOT} \sep$ $\tikzrom{mark2}$ $\quad \quad \quad $ $\tikzrom{mar2}$
           $\q \oasg \oper{H} \sep$
           $\q,\q' \oasg \oper{CNOT} \sep$
           $\q \oasg \oper{T} \sep$
           $\cons(1) \sep$
           $\q \oasg \oper{H} \sep$ $\tikzrom{mark3}$
           $\x \asg \meas(\q)$ $\quad \quad \quad $ $\tikzrom{mar3}$
         $\}$
\end{lstlisting}
\AddNotecopy{mark1}{mark3}{mark2}{$\cmd_0$}
\AddNotecopy{mar1}{mar3}{mar2}{$\cmd$}
\end{figure}
\begin{figure}[h]
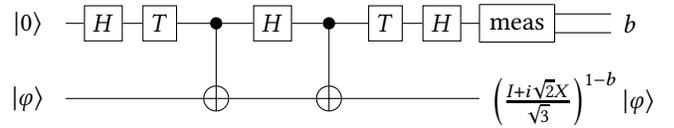

  \cstikz{rus.tikz}
  \caption{A quantum circuit illustrating a repeat-unitl-success pattern.}
  \label{fig:rus}
\end{figure}

Repeat-until-success \cite{RUS} can be used to implement quantum
unitary operators by using repeated measurements. An advantage of this approach
is that this often allows us to implement quantum unitary operators by using
fewer $T$ gates, which are costly to implement fault-tolerantly~\cite{BK05,GKMR14}.


  Consider the example in~\Cref{ex2}. This quantum algorithm will repeatedly execute the quantum operations described by the quantum circuit in Figure \ref{fig:rus}, as specified in \cite[Figure 8]{RUS}, where the operators $\oper{T}$ and $\oper{CNOT}$ correspond to the following quantum gates:
\[
T=
\begin{pmatrix}1 & 0 \\
0 & e^{i\frac{\pi}{4}}
\end{pmatrix}
\quad CNOT=
\begin{pmatrix}
1 & 0 & 0 &0 \\
0 & 1 & 0 & 0\\
0 & 0 & 0 & 1\\
0 & 0 & 1 & 0\\
\end{pmatrix} .
\]
  After measuring the first (ancilla) qubit $\q$, there are two possibilities. With probability $1/4$  , we measure $1$ and then the state of the second qubit $\q'$ is again $\ket \varphi$ and we repeat the algorithm.
  With probability $3/4$, we measure $0$ and then the algorithm terminates and the second qubit is now in state $\frac{I + i \sqrt 2 X}{\sqrt 3} \ket \varphi$.

\paragraph*{Analysis}
Let $\cmd \triangleq \cmd_0 \sep \x \asg \meas(\q)$ be the body of the while loop statement in the above program. For any classical (cost) expectation $\kappa \in (\Rext)^\memst$, by rules of \Cref{fig:qet},
it holds that
\[
  \qect{\cmd_0}{\kappa} = \kappa + \underline{2}
  .
\]
Indeed, one can check easily that each qubit operation in $\cmd_0$ leaves the expectation $\kappa$ unchanged. For example, we have that $\qect{\q \oasg \oper{H}}{\kappa} = \kappa[\oper{H}]=\kappa$. Therefore, we just have to take into account the two $\cons(1)$ statements.

As the probability of measuring $0$ is constant for each iteration ($\frac{3}{4}$, see \cite{RUS}), for any cost expectation $\kappa$, it holds that:
\begin{align*}
\qect{\x \asg \meas(\q)}{\kappa} &= \kappa[\x := 0;\oper{M}_0^\q] \up{\frac{3}{4}} \kappa[\x := 1;\oper{M}_1^\q]\\
&=\underline{\sfrac{3}{4}} \ucdot \kappa[\x := 0] + \underline{\sfrac{1}{4}} \ucdot \kappa[\x := 1]
\end{align*}
Putting this all together, the expectation of $\cmd$ is:
\begin{align*}
\qect{\cmd}{\kappa} &= \qect{\cmd_0}{\qect{\x \asg \meas(\q)}{\kappa}}\\
&= \qect{\x \asg \meas(\q)}{\kappa} + \underline{2}\\
&=\underline{\sfrac{3}{4}} \ucdot \kappa[\x := 0] + \underline{\sfrac{1}{4}} \ucdot \kappa[\x := 1]+\underline{2} .
\end{align*}
By Law \ref{idents:ui} of \Cref{fig:idents}, it suffices to find an expectation $\kappa$ satisfying the following inequalities
\begin{align}
\sem{\neg \x} \cdot \underline{0} &\leq \kappa \label{a1}\\
\sem{\x} \cdot \qect{\cmd}{\kappa} & \leq \kappa  \label{a2}
\end{align}
in order to compute (an upper bound on) the cost expectation $ \qect{\while(\x) \{ \cmd\}}{\underline{0}}$. Inequalities \eqref{a1} and \eqref{a2} are satisfied by setting
$\kappa \triangleq \sem{\x}\ucdot \underline{\frac{8}{3}}$.

We conclude by computing the expectation of the whole program
\begin{align*}
\qect{RUS(\q')}{\underline{0}} & \leq \qect{\x^\Bool \asg \true }{\kappa} \\
&= \kappa [\x := 1] = \underline{\sfrac{8}{3}}.
\end{align*}
The expected cost (the expected number of $T$ gates used) of this algorithm is bounded by $\frac{8}{3}$.
Note that this bound is tight.

\subsection{Chain of $k$ entangled qubits}\label{ex:2}
The following example illustrates that the presented cost analysis can also deal with nested while loops on a non-trivial example using classical data.
We consider a simple algorithm attempting to prepare a large entangled state, namely a graph state represented by a path on $k$ qubits, i.e. $\ket{\phi_k} = \prod_{i=0}^{k-2}CZ_{i,i+1}\otimes_{j=0}^{k-1} \ket +$ where $CZ$ is the following 2-qubit unitary transformation: 
$$CZ=
\begin{pmatrix}
1 & 0 & 0 &0 \\
0 & 1 & 0 & 0\\
0 & 0 & 1 & 0\\
0 & 0 & 0 & -1\\
\end{pmatrix}.
$$
One can prepare the desired state by initializing $k$ qubits in a row in the state $\ket +$ and then applying the $CZ$ gate $k-1$ times, one for each pair of consecutive qubits. Notice that the order in which the CZ are applied is irrelevant (as they are commuting). However, in some settings, like linear optical quantum computing, CZ cannot be implemented deterministically. Nielsen \cite{PhysRevLett.93.040503} showed that CZ can be implemented in linear optics with probability of success $1/4$. The case of a failure corresponds to a measurement of the corresponding two qubits. 

\begin{figure}
\begin{lstlisting}[caption={Applying a $CZ$ gate to $\q$, $\q'$ with probability $\sfrac{1}{4}$.}, label={l:fuse},style=qwhile]
$FUSE(\q,\q',\x) \triangleq$
  $\cons(1) \sep$
  $\ab^\Qubits \asg \ket{+}\sep$
  $\ba^\Qubits \asg \ket{+}\sep$
  $\x \asg \meas(\ab) \sep$
  $\y \asg \meas(\ba) \sep$
  $\ifa (\x \wedge \y)\{$ // with probability $\sfrac{1}{4}$
    $\q,\q' \oasg \oper{CZ} \sep$
    $\x \asg \true$ // Flag set to success
  $\}\elsea\{$ // with probability $\sfrac{3}{4}$
    $\x \asg \meas(\q)\sep$
    $\x \asg \meas(\q')\sep$
    $\x \asg \false$ // Flag set to failure
  $\}$
\end{lstlisting}
\begin{lstlisting}[caption={Chaining four qubits.},label={l:fuse-4}, style=qwhile]
$CHAIN4(\q_1,\q_2,\q_3,\q_4) \triangleq$
  $\x^\Bool \asg \false \sep$
  $\while(\neg \x)\{$
    $\x_1^\Bool \asg \false\sep$ // Left pair   $\tikzmark{c-mark11}$
    $\while(\neg \x_1)\{$
      $\q_1 \asg \ket{+} \sep$
      $\q_2 \asg \ket{+} \sep$
      $FUSE(\q_1, \q_2,\x_1)$
    $\}\sep$                                   $\tikzmark{c-mark12}$
    $\x_2^\Bool \asg \false \sep$ // Right pair $\tikzmark{c-mark21}$
    $\while(\neg \x_2)\{$
      $\q_3 \asg \ket{+} \sep$
      $\q_4 \asg \ket{+} \sep$
      $FUSE(\q_3, \q_4,\x_2)$
    $\}\sep$                                   $\tikzmark{c-mark22}$
    $FUSE(\q_2, \q_3,\x)$ // Fusion of pairs
  $\}$
\end{lstlisting}
\begin{lstlisting}[caption={Create a chain of $k$ entangled qubits.}, label={l:chain}, style=qwhile]
$CHAIN(k,\q_0,\ldots,\q_{k+3}) \triangleq$
  $\at^\Var \asg 0 \sep$
  $\q_0 \asg \ket{+} \sep$
  $\x^\Bool \asg \false \sep$
  $\while(0 \leq \at \wedge \at < k)\{$
    $CHAIN4(\q_{\at+1},\q_{\at+2},\q_{\at+3},\q_{\at+4})\sep$        $\tikzmark{c-mark1}$
    $FUSE(\q_\at, \q_{\at+1},\x)\sep$
    $\ifa (\x)\{ \at \asg \at+4\} \elsea \{ \at \asg \at-1\}\sep$
    $\ifa (\at = -1)\{ \at \asg 0 \sep \q_0 \asg \ket{+} \} \elsea \{\skp\}$ $\tikzmark{c-mark2}$
  $\}$
\end{lstlisting}
\AddNotecopy{c-mark11}{c-mark12}{c-mark12}{$\cmd_1$}%
\AddNotecopy{c-mark21}{c-mark22}{c-mark22}{$\cmd_2$}%
\AddNotecopy{c-mark1}{c-mark2}{c-mark2}{$\cmd$}%
\vspace{-5mm}
\end{figure}

\begin{notations}
In this example, we use $F(\overline{\q},\overline{\x})$ as a shorthand notation for a (non-recursive) call to program $F$ with parameters $\overline{\q},\overline{\x}$. A call to $F(\overline{\q},\overline{\x})$ consists in unfolding the statement of $F$ after a careful variable renaming, avoiding name clashes.
We also assume that, for a given sequence of qubits $\q_1,\ldots,\q_n$, we can access the $i$-th qubit through a call $\q_\x$, provided that variable $\x$ holds the value $i$ in the store.
\end{notations}

The program $FUSE(\q,\q',\x)$ in~\Cref{l:fuse} models the entanglement of two input qubits $\q$ and $\q'$ in state $\ket +$ with probability $\sfrac{1}{4}$. The Boolean variable $\x$ records whether this operation succeeded.
The $CHAIN4(\q_1,\q_2,\q_3,\q_4)$ from \Cref{l:chain} entangles four
given qubits, by iterating $FUSE$ until eventually all links have been established sucessfully.
The general algorithm in \Cref{l:chain} then makes use of this procedure by iteratively appending 4-entangled-qubits chains to the main chain, resulting eventually in a chain of $k \leq t \leq k + 3$ entangled qubits.

\paragraph*{Analysis}
Let us first consider the sub-program $FUSE$ from \Cref{l:fuse}.
Using laws of Figure~\ref{fig:qet}, it is not difficult to see, that for any classical expectation $\kappa$,
\begin{multline*}
  \qect{FUSE(\q,\q',\x)}{\kappa} = \\
  \underline{1} \up{} \underline{\sfrac{1}{4}}\ucdot\kappa[\x := \true] \up{} \underline{\sfrac{3}{4}} \ucdot \kappa[\x := \false]
  .
\end{multline*}
This can be formally verified by unfolding definitions, exploiting that $\kappa$ does not
depend on the quantum state.
Let us turn our attention to $CHAIN4$ from \Cref{l:fuse-4},
and observe
\[
  \qect{\cmd_2}{\kappa} \leq \underline{4} + \kappa
  ,
\]
when $\kappa$ is independent of $\x_2$ assigned in $\cmd_2$, i.e.,
$\kappa[\x_2 := b] = \kappa$.
To see this, take $g \triangleq \sem{\neg \x_2} \ucdot \underline{4} + \kappa$,
and hence
\begin{align*}
  \mparbox{3mm}{\qect{
    \q_3 \asg \ket{+} \sep \q_4 \asg \ket{+} \sep \\
    FUSE(\q_3, \q_4,\x_1)
  }{g}}
  &\\
  & = \qect{\q_3 \asg \ket{+} \sep \q_4 \asg \ket{+}}{\underline{1} + \underline{\sfrac{1}{4}} \ucdot g[\x_2 := \true]\\ \quad \up{} \underline{\sfrac{3}{4}}\ucdot g[\x_2 := \false]} \\
  & = \underline{4} + \kappa
    .
\end{align*}
Now it is clear that, since $\sem{\neg \x_2} \ucdot \underline{4} + \kappa \leq \sem{\neg \x_2} \ucdot \underline{4} + \kappa$ holds trivially,
$g$ constitutes an upper-invariant of the while loop in $\cmd_2$, see Law~\ref{idents:ui}.
Substituting $\false$ for $\x_2$
gives the bound for $\cmd_2$.
The same argument shows that
\[
  \qect{\cmd_1}{\kappa} \leq \underline{4} + \kappa
  ,
\]
for any $\kappa$ independent of $\x_1$.

Concerning the outer loop, let $\kappa$ now refer to a classical cost function independent
of the Boolean variables $\x_1$ and $\x_2$.
Putting things together,
\begin{multline*}
  \qect{ \cmd_1 \sep \cmd_2 \sep FUSE(\q_2,\q_3,\x)}{\sem{\neg \x} \ucdot \underline{36} + \kappa} \\
  \leq \underline{4} \up{} \underline{4} \up{} \underline{28} + \kappa = \underline{36} + \kappa
  ,
\end{multline*}
and finally
\[
  \qect{CHAIN4(\q_1,\q_2,\q_3,\q_4)}{\kappa} \leq \underline{36} \up{} \kappa
  ,
\]
via Law~\ref{idents:ui}.
Concerning the overall code from \Cref{l:chain}, let us now define the classical expectation
\[
  f \triangleq \sem{0 \leq \at \wedge t < k + 4} \ucdot \underline{148} \ucdot (\underline{k} \ud \sem{\at} \up{} \underline{4}).
\]
We obtain
\begin{align*}
   \mparbox{2mm}{\qect{\cmd}{f}}\\
   & = \underline{36} \up{} \underline{1}  \up{} \underline{\sfrac{1}{4}} \ucdot \sem{\at \not= -1} \ucdot f[\at := \at + 4]  \\
  & \quad  \up{}  \underline{\sfrac{3}{4}} \ucdot ( \sem{\at = -1} \ucdot f[\at := 0] \up{} \sem{\at \not= -1} \ucdot f[\at := \at - 1] ) \\
  & = \underline{37} \up{} \sem{0 \leq \at \wedge \at < k+4} \ucdot  (\underline{\sfrac{1}{4}} \ucdot(\underline{148} \ucdot (\underline{k} \ud \sem{\at}) ) \\
  &\quad \up{}  \underline{\sfrac{3}{4}} \ucdot \underline{148} \ucdot (\underline{k} \ud \sem{\at} \up{} \underline{5} )).
\end{align*}
Exploiting the loop-guard $0 \leq \at < k$, we finally
establish that $f$ is an upper-bound to the expected runtime of the loop,
where the required inequality is in particular
encompassed by the inequality
{\small
\begin{multline*}
  \sem{0 \leq \at \wedge \at < k} \cdot \underline{37} \up{} \underline{37} \ucdot (\underline{k} \ud \sem{\at}) \up{} \underline{111}  \ucdot (\underline{k} \ud \sem{\at} \up{} \underline{5})\\
  \leq \underline{148} \ucdot (\underline{k} \ud \sem{\at} \up{} \underline{4})
  ,
\end{multline*}
} 
which can be easily seen to hold.
Substituting $\underline{0}$ for $\sem{\at}$ in this expectation,
we conclude that the overall expected cost is bounded by $148 \times (k+4)$.

\subsection{Quantum walk}\label{ex:3}
In this last example, we consider the Hadamard quantum  walk on an $n$-circle as defined in ~\cite[Section 6.2]{LZY19}.
Our goal is to illustrate on a non-trivial example that the cost analysis may depend directly on the program quantum state, as in \Cref{ex:cointoss}.

Let $\q$ be a quantum bit of the 2-dimensional state space $\mathcal{H}_\q$, whose basis states $\ket{L}$ and $\ket{R}$ indicate directions Left and Right, respectively. Let $\mathcal{H}_\p$ be an $n$-dimensional Hilbert space with orthonormal basis $\ket{0}, \ket{1}, \ldots, \ket{n-1}$ for the positions. The state space $\mathcal{H}$ for the quantum  walk is defined by $\mathcal{H} \triangleq  \mathcal{H}_\q \otimes \mathcal{H}_\p$. The program itself is given in \Cref{lst:qrw}.
\begin{figure}
\begin{lstlisting}[caption={Quantum  walk.}, label={lst:qrw}, style=qwhile]
  $\x^\Bool \asg \true \sep$
  $\while(\x) \{ $
    $\x \asg \meas(\p)\sep$ $\tikzmark{truc1}$
    $\q \oasg \oper{H} \sep$
    $\q, \p \oasg \oper{S} \sep$
    $\cons(1) $ $\tikzmark{truc2}$
  $\}$
\end{lstlisting}
\end{figure}
The operator $\oper{S}$ shifts the position depending of the direction state and is defined by the following standard unitary operator:
$S = \Sigma_{i=0}^{n-1} \ket{L}\bra{L} \otimes \ket{i \ominus 1}\bra{i} + \Sigma_{i=0}^{n-1} \ket{R}\bra{R} \otimes \ket{i \oplus 1}\bra{i}$, where $\oplus$ and $\ominus$ denote addition and substraction modulo $n$.

\paragraph*{Adaptation}
We need to adapt slightly the operational semantics and expected cost transformer to this particular setting.
Any quantum state $\ket{\varphi}$ can be written as $$\ket{\varphi}  \triangleq  \sum_{i = 0}^{n-1} a_i\ket{L}\ket{i}+\sum_{i = n}^{2n-1} a_i\ket{R}\ket{i-n}=\begin{pmatrix}
a_0\\
\vdots\\
a_{2n-1}
\end{pmatrix},$$
for $n \geq 1$ and for $a_i \in \mathbb{C}$ such that $\Sigma_i \size{a_i}^2=1$. The probability that the quantum state $\ket{\varphi}$ is at position $0$ is given by $p_0^\p\ket{\varphi} \triangleq \bra{\varphi} (I_2 \otimes \ket{0}\bra{0}) \ket{\varphi}$. The probability that the quantum state $\ket{\varphi}$ is at a  position distinct from $0$ is given by $p_{\neq 0}^\p\ket{\varphi} \triangleq \bra{\varphi} (I_{2n}-I_2 \otimes \ket{0}\bra{0}) \ket{\varphi}$. These two probabilities trivially satisfy $p_0^\p+p_{\neq 0}^\p=\underline{1}$ and it holds that $p_{\neq 0}\ket{\varphi}= 1- (\size{a_0}^2+ \size{a_n}^2) = \sum_{i \neq 0,i \neq n} \size{a_i}^2$.

We adapt in a direct and obvious way the result of the calculation of a measurement to the $n$-dimensional case: the result of measuring the outcome $0$ ($\false$), $\oper{M}_0^\p$, and the result of measuring an outcome distinct from $0$ ($\true$), $\oper{M}_{\neq 0}^\p$, are defined by $\oper{M}_0^\p \triangleq \frac{1}{\sqrt{p_{0}^\p}}I_2 \otimes \ket{0}\bra{0}$ and $\oper{M}_{\neq 0}^\p \triangleq \frac{1}{\sqrt{p_{\neq 0}^\p}} (I_{2n}-I_2 \otimes \ket{0}\bra{0})$, respectively.

The operational semantics of \Cref{fig:os} and the expected cost transformer of \Cref{fig:qet} can be adapted straightforwardly to this new setting. E.g., the rule of \Cref{fig:qet} for measurement on qubit $\p$ is rewritten as follows:
\[
\qet{\x \asg \meas(\p)}{f} = f[\x := 0;\oper{M}_0^\q] \up{p_0^\q} f[\x := 1;\oper{M}_{\neq 0}^\q].
\]
All the other rules remain unchanged. Our soundness results still hold in this context. In particular, \Cref{thm:idents} and \Cref{c:ect-soundness} are still valid.

\paragraph*{Analysis}
In order to analyse the expected cost of the above program, we search for an expectation $g$ satisfying the prerequisite for applying Law~\ref{idents:ui}:
\begin{equation}\label{equ:1}
 \sem{\x} \ucdot \qect{
    \x \asg \meas(\p)\sep\\
    \q \oasg \oper{H} \sep\\
    \q, \p \oasg \oper{S} \sep\\
    \cons(1) \\
  }{g} \leq g
    .
\end{equation}


Using the (adapted) laws of Figure~\ref{fig:qet}, the following equalities can be derived:
\begin{align*}
\qect{ \cons(1)}{g} &= \underline{1} +g \triangleq g_1 \\
\qect{ \q, \p \oasg \oper{S}_{\q,\p} }{g_1}&= (\underline{1} +g)[\oper{S}_{\q,\p}] \triangleq g_2\\
\qect{ \q \oasg \oper{H}_\q}{g_2}&= g_2[\oper{H}_\q]  =(\underline{1} +g)[\oper{S}_{\q,\p}\oper{H}_\q]\triangleq g_3\\
\qect{ \x \asg \meas(\p)}{g_3} &=g_3[\x := 0;\oper{M}_0^\p]\up{p_0^\p} g_3[\x := 1;\oper{M}_{\neq 0}^\p]
\end{align*}
where the matrices corresponding to the operators $\oper{S}_{\q,\p}$ and $\oper{H}_\q$ are equal to
$H \otimes I_n
 \text{ and }S
$, respectively.


In a nutshell, $\qect{ \x \asg \meas(\p)}{g_3}$ can be written as:
\begin{align*}
\qect{ \x \asg \meas(\p)}{g_3} &= p_{0}^\p\cdot (\underline{1} +g)[\oper{S}_{\q,\p}\oper{H}_\q][\x := 0;\oper{M}_{0}^\p] \\
&\quad + p_{\neq 0}^\p\cdot (\underline{1} +g)[\oper{S}_{\q,\p}\oper{H}_\q][\x  := 1;\oper{M}_{\neq 0}^\p]
\end{align*}
and \Cref{equ:1} can be simplified as follows:
\begin{equation}\label{equ}
p_{\neq 0}^\p (\underline{1} +g)[\x := 1; \oper{S}_{\q,\p}\oper{H}_\q \oper{M}_{\neq 0}^\p] \leq g.
\end{equation}


Equation $(\ref{equ})$ holds if for any store $s$ and any quantum state $\ket{\varphi}$, we have:
$$
p_{\neq 0}^\p\ket{\varphi} (1 +g \ltuple s[\x := 1],\oper{S}_{\q,\p}\oper{H}_\q \oper{M}_{\neq 0}^\p\ket{\varphi}\rtuple)\leq g(s, \ket{\varphi}).
$$

Notice that $\ket{\psi} \triangleq \oper{S}_{\q,\p}\oper{H}_\q \oper{M}_{\neq 0}^\p\ket{\varphi}$ will be the quantum state entering the loop in the next iteration.

For any complex number $z \in \mathbb{C}$, $z^*$ will denote the complex conjugate of $z$ and the real part of $z$ will be denoted by $\mathfrak{R}(z)$.

Consider the expectation $g_n$ defined below.

\begin{align*}
g_n\begin{pmatrix}
a_0\\
\vdots\\
a_{2n-1}
\end{pmatrix} &\triangleq \sum_{i=0}^{2n-1}f_n(i) \size{a_i}^2 + 2 \sum_{j=0}^{2n-1} \sum_{k=0}^{2n-1}h_n(j,k)\mathfrak{R}(a_j a_k^*)
\end{align*}
with $f_n(i) \triangleq \begin{cases}i(n-i)+1 & \text{if } 0 \leq i \leq n-1 \\ (i-n)(2n-i)+1 & \text{if } n \leq i \leq 2n-1 \end{cases}$

and with $h_n(j,k)\triangleq$
$$\begin{cases}(-1)^{\frac{j-k}{2}}k(n-1-j) & \text{ if } \begin{cases}k,j \in [0,n-1], \\ k < j,  \\  (j - k)\ \%\ 2 = 0\end{cases}\\ (-1)^{\frac{j-k}{2}}(2n-j)(k-n-1)& \text{ if } \begin{cases}k,j \in [n,2n-1],\\ k <j,  \\ (j - k) \ \%\ 2 = 0\end{cases}\\  (-1)^{\frac{k-j-n}{2}}(j+k-2n)& \text{ if } \begin{cases} j+n, k \in [n+1,2n-1], \\j+n \leq k,\\ (k - (j+n)) \ \%\ 2 = 0\end{cases}\\ 0 & \text{ otherwise.} \end{cases}$$

$g_n$ is a solution of the inequality in Equation~(\ref{equ}). This can be shown by symbolically computing the substraction of the left-hand side and the right-hand side of the inequality.

Now we consider the simple case where $n=2$. For $s$ such that $\sem{\x}^s=1$, Equation~(\ref{equ}) can be rewritten as:
\begin{equation*}
(\size{a_1}^2+ \size{a_3}^2)(1+g\ltuple s,\frac{1}{\sqrt{2(\size{a_1}^2+ \size{a_3}^2)}} \begin{pmatrix}
a_1+a_3\\
0 \\
a_1 -a_3\\
0 \\
\end{pmatrix} \rtuple ) \leq g\ltuple s,\begin{pmatrix}
a_0\\
a_1 \\
a_2 \\
a_3 \\
\end{pmatrix} \rtuple.
\end{equation*}

The above inequality is satisfied for $g_2 \ltuple s, \ket{\varphi} \rtuple \triangleq 1+\size{a_1}^2+\size{a_3}^2$.
Hence, starting in position $\ket{1}$ (i.e., $a_1\ket{L}\ket{1}+a_3\ket{R}\ket{1}$, with $\size{a_1}^2+\size{a_3}^2=1$), the expected cost is $2$, whereas starting in position $\ket{0}$, the expected cost is $1$ (as $a_1=a_3=0$).

In the case, where $n=3$, the expectation $g_3 \ltuple s, \ket{\varphi} \rtuple \triangleq 
2-\size{a_0}^2-\size{a_3}^2+\size{a_2+a_5}^2+\size{a_1-a_4}^2$ is a solution to Equation~(\ref{equ}).

Note that our expectations $g_n$ can be recovered from the matrices $Q_n$ in the work of~\cite[Section 6.2]{LZY19} as follows:  $g_n\ltuple s,\ket{\varphi} \rtuple \triangleq \bra{\varphi}Q_n\ket{\varphi}$.
%

\section{Conclusion and Future Work}

We presented an adequate notion of quantum expectation transformer and showed through practical examples that it can be used to infer upper bounds on the expected cost of quantum programs.

As already indicated, our \textsf{qet}-calculus provides a principled foundation for automation. While this problem is clearly undecidable in general,
a restriction to a well-defined function space for expectations may allow for an effective solution, at the price of incompleteness. Existing
work in the literature on automation of expected cost transfomers or related work for \emph{classical} programs cf.~\cite{NCH18,AMS20,WFGCQS:PLDI:19,MeyerHG21} should provide ample guidance in this respect.

\bibliography{references}

\appendix
\newcommand{\fl}[1]{\mparbox{3mm}{#1}}
\newcommand{\cmt}[1]{\text{\quad(\emph{#1})}}

\newpage
\section{Background Material}

We state some intermediate properties that follow from the Kegelspitze structure~\cite{kegelspitzen}.
\begin{proposition}\label{p:lfp}
  In any Kegelspitze $K$, any continuous map
  $\chi : K \to K$ has a least fixed-point
  given by
  \[
    \lfp \chi = \sup_{n \in \mathbb{N}} \chi^n(\bot) .
  \]
  This operator $\lfp : (K \to K) \to (K \to K)$ is
  itself continuous.
\end{proposition}

\begin{lemma}\label{l:dsum-bari}
  Let $f,g : A \to K$ for Kegelspitze $K$ and $p \in [0,1]$.
  Then
  \[
    \E{\delta}{f +_p g} = \E{\delta}{f} +_p \E{\delta}{g}
    .
  \]
\end{lemma}

\begin{lemma}\label{l:dsum-exp}
  Let $f : A \to K$ for Kegelspitze $K$.
  Then
  \[
    \E{\sum_{i \in I} p_i \cdot \delta_i}{f} = \sum_{i \in I} p_i \cdot \E{\delta_i}{f}
    .
  \]
\end{lemma}

\begin{lemma}\label{l:exp-cont}
  Let $f : A \to K$ for Kegelspitze $K$.
  For every $\omega$-chain $(d_i)_i$ of sub-distributions,
  \[
    \sup_i \E{\delta_i}{f} = \E{\sup_i \delta_i}{f}
    .
  \]
\end{lemma}

\begin{lemma}\label{l:pars-tomulti}
  For any PARS $\to$ over $A$ and $\delta,\epsilon \in \mathcal{D}(A)$,
  if
  \[
    \delta \tomulti{c}_n \epsilon
  \]
  then
  (i)~$\delta = \lmulti p_i : a_i \rmulti_{i \in I}$,
  (ii)~$\epsilon = \sum_{i \in I} p_i \cdot \epsilon_i$,
  (iii)~$\lmulti 1 : a_i \rmulti \tomulti{c_i}_n \epsilon_i$, and
  (iv)~$c = \sum_{i \in I} c_i$.
\end{lemma}
\begin{proof}
  One first shows the lemma for the one-step reduction relation $\tomulti{\cdot}$, by induction on the derivation of $\delta \tomulti{c} \epsilon$.
  This then generalises to $\tomulti{\cdot}_n$ by induction on $n$.
\end{proof}

\section{Proof of \Cref{l:QET}}
\label{app:QET}


Recall that, for
$\cmd \in \Cmds$, $f \in \CSd^\memst$ and $\sigma \in \memst$,
$\QET{\cmd}{f}(\sigma) = \QETT{f}(\cmd,\sigma)$,
where $\QETT{f} \triangleq \lfp(\xi_f)$
for
\[
  \xi_f(F) \triangleq \lambda \tau.
  \begin{cases}
    f(\tau) & \text{if $\tau \in \memst$, } \\
    c \CSp \E{\delta}{F} & \text{if $\tau \in \conf$ and $\tau \toop{c} \delta$.}
  \end{cases}
\]
We first prove continuity, thus in particular well-definedness of the $\QET{-}{-}$ function:
\begin{lemma}[Continuity and Monotonicity]\label{l:QET:continuity}
For any $\omega$-chain $(f_n)_{n \in \N}$ and any functions $f,g$ s.t. $f \leq g$ the following hold:
  \begin{enumerateenv}
  \item\label{l:QET:continuity:cont} $\sup_n \QET{\cmd}{f_n} = \QET{\cmd}{\sup_n f_n}$
    ; and
  \item\label{l:QET:continuity:mono} $\QET{\cmd}{f} \leq \QET{\cmd}{g}$.
  \end{enumerateenv}
\end{lemma}
\begin{proof}
  Note that $\xi_f$ as defined just above is continuous, since in particular
  $\CSp$ is continuous in both its arguments, and $\E{\delta}{-}$ is continuous.
  Concerning the latter, it is actually sufficient to restrict to cases of distributions $\delta$
  with $\mu \toop{c} \delta$ for configuration $\mu$.
  Then, by definition, either $\delta = \lmulti 1 : \nu \rmulti$ or $\E{\delta}{f} = f(\nu)$
  or $\delta = \lmulti p : \nu_1, 1 - p : \nu_2 \rmulti$ and $\E{\delta}{f} = f(\nu_1) +_p f(\nu_2)$.
  In either case, continuity of $\E{\delta}{f}$ follows from continuity of $+_p$.
  Thus \eqref{l:QET:continuity:cont} is a consequence of \Cref{p:lfp},
  from which then also \eqref{l:QET:continuity:cont} follows.
\end{proof}

To prove \Cref{l:QET}, we will reason via approximations of the involved functions.
To this end, for $n \in \N$ and $f \in \CSd^\memst$
let us define $\QETT[n]{f} : \memst \to \CSd$
as the $n$-th approximant $\xi_f^n(\bot)$, thus, by \Cref{p:lfp},
\begin{align*}
  \QET{\cmd}{f}(\sigma)
  & = \left(\sup_{n \in \N} \QETT[n]{f}\right)(\cmd,\sigma) \\
  & = \sup_{n \in \N} \QETT[n]{f}(\cmd,\sigma)
    .
\end{align*}
For cost structure $(\CSd,\CSp)$, we define, for $n \in \N$,
$\ecost[n]_\to(-) : \conf \to \Rext$  and
$\nf[n]_\to(-) : \conf \to \mathcal{D}(\memst)$
inductively by setting
\begin{align*}
  \ecost[0]_\to(\mu) & = 0 \\
  \ecost[n+1]_\to(\mu) & =
  \begin{cases}
    0 & \text{if $\mu \in \memst$,} \\
    c \CSp \sum_{i \in I} p_i \cdot \ecost[n]_\to(\nu_i) & \text{if $\mu \toop{c} \lmulti p_i : \nu_i \rmulti_{i \in I}$;} \\
  \end{cases}
  \\[3mm]
  \nf[0]_\to(\mu) & = \varnothing \\
  \nf[n+1]_\to(\mu) & =
  \begin{cases}
   \lmulti 1 : \mu \rmulti & \text{if $\mu \in \memst$,} \\
   \sum_{i \in I} p_i \cdot \nf[n]_\to(\nu_i) & \text{if $\mu \toop{c} \lmulti p_i : \nu_i \rmulti_{i \in I}$.}
  \end{cases}
\end{align*}
These approximate $\ecost_\to$ and $\nf_\to$ defined in \Cref{sec:os}, respectively,
in the following sense:
\begin{lemma}\label{l:rew-approx}
  Let $\mu \in \conf \cup \memst$.
  Then
  \begin{enumerate}
  \item\label{l:rew-approx:ecost} $\ecost_\to (\mu) = \sup_{n \in \N} \ecost[n]_\to(\mu)$, and
  \item\label{l:rew-approx:nf} $\E{\nf_\cmd(\mu)}{f} = \sup_{n \in \N} \E{\nf[n]_\to(\mu)}{f}$.
  \end{enumerate}
\end{lemma}
\begin{proof}
  Concerning Property~\eqref{l:rew-approx:ecost}, we show that
  \[
    \lmulti 1 : \mu \rmulti \tomulti{c}_n \epsilon
    ,
  \]
  iff $\ecost[n]_\to(\mu) = c$, from which the property follows then by definition of $\ecost_\to$.
  The proof is by induction on $n$. The base case is trivial, let us consider the inductive step
  where
  \[
    \lmulti 1 : \mu \rmulti \tomulti{c} \delta \tomulti{d}_n \epsilon
  \]
  where we have to show $\ecost[n+1](\mu) = c + d$.
  Applying \Cref{l:pars-tomulti}, for
  $\delta = \lmulti p_i : \nu_i \rmulti_{i \in I}$ we see that $d = \sum_{i \in I} d_i$
  where for all $i \in I$, $\lmulti 1 : \nu_i \rmulti \tomulti{d_i}_n \epsilon_i$ for some $\epsilon_i$.
  Summing up
  \begin{align*}
    \ecost[n+1]_\to(\mu) & = c + \sum_i p_i \cdot \ecost[n]_\to(\nu_i) \\
    & \cmt{induction hypothesis} \\
    & = c + \sum_i p_i \cdot d_i = c + d
    .
  \end{align*}

  Concerning Property~\eqref{l:rew-approx:nf}, we show first that
  \[
    \lmulti 1 : \mu \rmulti \tomulti{c}_n \epsilon
    ,
  \]
  iff $\nf[n+1]_\to(\mu) = \epsilon \restrict_{term}$.
  Again the proof is by induction on $n$. In the base case $n = 0$,
  we consider $\lmulti 1 : \mu \rmulti \tomulti{0}_0 \lmulti 1 : \mu \rmulti$.
  If $\mu \in \memst$ then $\nf[1]_\to(\mu) = \lmulti 1 : \mu \rmulti = \lmulti 1 : \mu \rmulti \restrict_{term}$
  as desired, in the case $\mu \in \conf$ with $\mu \toop{c} \{ p_i : \nu_i \}_{i \in I}$
  we have $\nf[1]_\to(\mu) = \sum_{i \in I} p_i \cdot \nf[0]_\to(\nu_i) = \sum_{i \in I} p_i \cdot \varnothing = \varnothing = \lmulti 1 : \mu \rmulti \restrict_{term}$.
  This concludes the base case, we move to the inductive case, where
  \[
    \lmulti 1 : \mu \rmulti \tomulti{c} \delta \tomulti{d}_n \epsilon
    .
  \]
  Applying \Cref{l:pars-tomulti}, for
  $\delta = \lmulti p_i : \nu_i \rmulti_{i \in I}$ we see that
  $\epsilon = \sum_{i \in I} p_i \cdot \epsilon_i$
  with $\lmulti 1 : \nu_i \rmulti \tomulti{d_i}_n \epsilon_i$ for all $i \in I$.
  Then
  \begin{align*}
    \nf[n+2]_\to(\mu)
    & = \sum_i p_i \cdot \nf[n+1]_\to(\nu_i) \\
    & \cmt{induction hypothesis} \\
    & = \textstyle\sum_i p_i \cdot (\epsilon_i \restrict_{term}) \\
    & = \left(\textstyle\sum_i p_i \cdot \epsilon_i\right) \restrict_{term} = \epsilon \restrict_{term}
    .
  \end{align*}

  From this, we finally conclude as Property~\eqref{l:rew-approx:nf}
  \begin{align*}
    \nf_\to(\mu) & = \sup_{n \in \N}\{ \delta \restrict_{term} \mid \lmulti 1: a\rmulti \tomulti{c}_n \delta \} \\
                 & = \sup_{n \in \N} \nf[n+1]_\to(\mu) = \sup_{n \in \N} \nf[n]_\to(\mu)
    . \qedhere
  \end{align*}
\end{proof}

\again{l:QET}
\begin{proof}
We show
  \[
    \QETT[n]{f}(\mu) = \ecost[n]_\to(\mu) \CSp \E{\nf[n]_\to(\mu)}{f}
    ,
  \]
for any $\mu \in \conf \cup \memst$,
since then we have
  \begin{align*}
    \fl{\QET{\cmd}{f}(\sigma)} \\
    & = \sup_{n \in \N} \QETT[n]{f}(\cmd,\sigma) \\
    & = \sup_{n \in \N} (\ecost[n]_\to(\cmd,\sigma) \CSp \E{\nf[n]_\to(\cmd,\sigma)}{f}) \\
    & \cmt{continuity of $\CSp$} \\
    & = \sup_{n \in \N} \ecost[n]_\to(\cmd,\sigma) \CSp \sup_{n \in \N}\E{\nf[n]_\to(\cmd,\sigma)}{f} \\
    & \cmt{\Cref{l:exp-cont}} \\
    & = \sup_{n \in \N} \ecost[n]_\to(\cmd,\sigma) \CSp \E{\sup_{n \in \N} \nf[n]_\to(\cmd,\sigma)}{f} \\
    & \cmt{\Cref{l:rew-approx}} \\
    & = \ecost_\to(\cmd,\sigma) \CSp \E{\nf_\to(\cmd,\sigma)}{f} \\
    & \cmt{by definitions} \\
    & = \ecost_\cmd(\sigma) \CSp \E{\nf_\cmd(\sigma)}{f}
      .
  \end{align*}
  Let us now prove the above inequality on approximants.
  The proof is by induction on
  $n$. Let $\mu \in \conf \cup \memst$.
  \begin{proofcases}
    \proofcase{n = 0}
    Trivially,
    \begin{align*}
      \QETT[0]{f}(\mu) = \bot
      & = 0 \CSp \bot \\
      & = \ecost[0]_\to(\mu) + \E{\varnothing}{f} \\
      & = \ecost[0]_\to(\mu) + \E{\nf[0]_\to(\mu)}{f}
        .
    \end{align*}
    \proofcase{$n+1$}
    Consider the step case.
    If $\mu \in \memst$,
    then
    \begin{align*}
      \QETT[n+1]{f}(\mu) = \bot
      & = 0 \CSp f(\mu)\ \\
      & = \ecost[n+1]_\to(\mu) + \E{\{1 : \mu \}}{f} \\
      & = \ecost[n+1]_\to(\mu) + \E{\nf[n+1]_\to(\mu)}{f}
        .
    \end{align*}
    Otherwise, if $\mu \in \conf$ then
    $\mu \toop{c} \lmulti p_i : \nu_i \rmulti_{i \in I}$
    and we have
    {
    \begin{align*}
      \fl{\QETT[n + 1]{f} (\mu)} \\
      & = c \CSp \sum_{i \in I} p_i \cdot \QETT[n]{f}(\nu_i) \\
      & \cmt{induction hypothesis} \\
      & = c \CSp \sum_{i \in I} p_i \cdot (\ecost[n]_\to(\nu_i) \CSp \E{\nf[n]_\to(\nu_i)}{f}) \\
      & \cmt{\Cref{d:cs}(3)} \\
      & = c \CSp \left(\sum_{i \in I} p_i \cdot \ecost[n]_\to(\nu_i) \CSp \sum_{i \in I} p_i \cdot \E{\nf[n]_\to(\nu_i)}{f} \right) \\
      & \cmt{\Cref{d:cs}(2)} \\
      & = \left(c \CSp \sum_{i \in I} p_i \cdot \ecost[n]_\to(\nu_i)\right) \CSp \sum_{i \in I} p_i \cdot \E{\nf[n]_\to(\nu_i)}{f}  \\
      & \cmt{\Cref{l:dsum-exp}} \\
      & = \left(c \CSp \sum_{i \in I} p_i \cdot \ecost[n]_\to(\nu_i)\right) \CSp \E{\sum_{i \in I} p_i \cdot \nf[n]_\to(\nu_i)}{f}  \\
      & \cmt{by definitions} \\
      & = \ecost[n+1]_\to(\mu) \CSp \E{\nf[n+1]_\to(\mu)}{f} .
        \qedhere
    \end{align*}
    }
  \end{proofcases}
\end{proof}

\section{Proof of \Cref{l:QET:main-props}}

\again{l:QET:main-props}
\begin{proof}
  \newcommand{\sepp}{\mathrel{;\!;}}
  We start with the proof of the first identity.
  For $\mu \in \conf \cup \memst$ and $\cmd \in \Cmds$, let us
  define $\mu  \sepp \cmd \in \conf$ by case analysis on $\mu$ as follows:
  \begin{align*}
    (\cmd_1,\sigma) \sepp \cmd_2 & \triangleq (\cmd_1 \sep \cmd_2; \sigma) \\
    \sigma \sepp \cmd_2 & \triangleq (\cmd_2; \sigma)
  \end{align*}
  By this notation, Rule (Seq) from \Cref{fig:os}
  instantiated to $\cmd_1;\cmd_2$,
  can be written as
  \[
    \begin{prooftree}
      \hypo{\ltuple\cmd_1 ,s,\ket{\varphi}\rtuple \toop{c} \lmulti p_i : \mu_i \rmulti_{i \in I}}
      \infer1[(Seq)]{\ltuple\cmd_1 \sep \cmd_2,s,\ket{\varphi}\rtuple \toop{c} \lmulti p_i: \mu_i\sepp\cmd_2 \rmulti_{i \in I}}
    \end{prooftree}
  \]
  Let $lhs$ and $rhs$ be the left- and right-hand side of the Identity \eqref{l:QET:main-props:seq}.
  We show $lhs \leq rhs$ and $rhs \leq lhs$ separately.
  \begin{proofcases}
    \proofcase{$lhs \leq rhs$}
    We prove
    \begin{multline}\label{l:QET:comp-approx:1}
      \QETT[n]{f}(\cmd_1 \sep \cmd_2,\sigma) \\
      \leq \QETT[n]{\lambda \tau. \QETT[n]{f}(\cmd_2,\tau)}(\cmd_1,\sigma)
      ,
    \end{multline}
    for all $n \in \N$ and $\sigma \in \memst$. From this, the case follows as $\QET{\cmd}{h} = \sup_n \QETT[n]{h}(\cmd,-)$,
    using continuity and monotonicity of the transformer (\Cref{l:QET:continuity}).
    The proof is by induction on $n$.
    The case $n = 0$ is trivial, as then $lhs = \underline{0} = rhs$.

    Concerning the inductive step, fix $\sigma \in \memst$.
    Note
    \begin{equation}
      \label{eq:qet-aux-ih:1}
      \QETT[n]{f}(\mu \sepp \cmd)
      \leq \QETT[n]{\lambda \tau. \QETT[n]{f}(\cmd,\tau)} (\mu)
      ,
    \end{equation}
    holds for any $\mu \in \conf \cup \memst$ and $\cmd \in \Cmds$.
    This follows by case analysis on $\mu$, using induction hypothesis in the case $\mu \in \conf$.
    Now suppose
    \[
      \ltuple\cmd_1 \sep \cmd_2,s,\ket{\varphi}\rtuple \toop{c} \lmulti p_i: \mu_i\sepp\cmd_2 \rmulti_{i \in I}
      ,
    \]
    since
    \[
      \ltuple\cmd_1 ,s,\ket{\varphi}\rtuple \toop{c} \lmulti p_i : \mu_i \rmulti_{i \in I}
      .
    \]
    We can thus conclude as
    \begin{align*}
      \fl{\QETT[n+1]{f}(\cmd_1 \sep \cmd_2, \sigma)} \\
      & \cmt{unfolding definition} \\
      & = c \CSp \sum_i p_i \cdot \QETT[n]{f}(\mu_i \sepp \cmd_2) \\
      & \cmt{Eq. \eqref{eq:qet-aux-ih:1} \& monotonicity of barycentric operations} \\
      & \leq c \CSp \sum_i p_i \cdot \QETT[n]{\lambda \tau. \QETT[n]{f}(\cmd_2,\tau)} (\mu_i) \\
      & \cmt{folding definition} \\
      & = \QETT[n+1]{\lambda \tau. \QETT[n]{f}(\cmd_2,\tau)}(\cmd_1,\sigma) \\
      & \cmt{$\QETT[n]{f}$ is monotone in $f$ and $n$} \\
      & = \QETT[n+1]{\lambda \tau. \QETT[n+1]{f}(\cmd_2,\tau)}(\cmd_1,\sigma)
        .
    \end{align*}
    \proofcase{$rhs \leq lhs$}
    It is sufficient to show
    \begin{multline}
      \label{l:QET:comp-approx:2}
      \QETT[n]{\lambda \tau. \QETT{f}(\cmd_2,\tau)}(\cmd_1,\sigma) \\
      \leq \QETT{f}(\cmd_1 \sep \cmd_2,\sigma)
      ,
    \end{multline}
    for all $n \in \N$.
    Suppose
    \[
      \ltuple\cmd_1 ,s,\ket{\varphi}\rtuple \toop{c} \lmulti p_i : \mu_i \rmulti_{i \in I}
      ,
    \]
    and thus
    \[
      \ltuple\cmd_1 \sep \cmd_2,s,\ket{\varphi}\rtuple \toop{c} \lmulti p_i: \mu_i\sepp\cmd_2 \rmulti_{i \in I}
      .
    \]
    The proof is by induction on $n$, as before, it is sufficient to consider
    the inductive step.
    Using induction hypothesis, we
    obtain
    \begin{equation}
      \label{eq:qet-aux-ih:2}
      \QETT[n]{\lambda \tau. \QETT{f}(\cmd,\tau)} (\mu)
      \leq \QETT{f}(\mu \sepp \cmd)
      ,
    \end{equation}
    for any $\mu \in \conf \cup \memst$ and $\cmd \in \Cmds$.
    We conclude then \eqref{l:QET:comp-approx:2} as
    \begin{align*}
      \fl{\QETT[n+1]{\lambda \tau. \QETT{f}(\cmd_2,\tau)}(\cmd_1,\sigma)} \\
      & \cmt{unfolding definition} \\
      & = c \CSp \sum_i p_i \cdot \QETT[n]{\lambda \tau. \QETT{f}(\cmd_2,\tau)}(\mu_i) \\
      & \cmt{Eq. \eqref{eq:qet-aux-ih:2} \& monotonicity of barycentric operations} \\
      & = c \CSp \sum_i p_i \cdot \QETT{f}(\mu_i \sepp \cmd_2) \\
      & \cmt{folding definition} \\
      & = \QETT{f}(\cmd_1\sep\cmd_2,\sigma)
        .
    \end{align*}
    This concludes the proof of the first identity.
  \end{proofcases}
  We now come to the second identity.  Let
  \[
    \chi_{f,\bexp}(F) \triangleq \QET{\cmd}{F} +_{\sem{\bexp}} f
  \]
  Again, we perform case analysis.
  \begin{proofcases}
    \proofcase{$lhs \leq rhs$}
    We prove the stronger statement
    \[
      \QETT[n]{f}(\while(\bexp)\{ \cmd \},\sigma)
      \leq \chi_{f,\bexp}^n(\bot)(\sigma)
      ,
    \]
    for all $n \in \N$ and $\sigma \in \memst$.
    It is sufficient to consider the inductive step.
    Fix $\sigma = (s,\ket{\varphi}) \in \memst$. Then
    \begin{align*}
      \fl{\QETT[n+1]{f}(\while(\bexp)\{ \cmd \},\sigma)} \\
      & =
      \begin{cases}
        \QETT[n]{f}(\cmd\sep\while(\bexp)\{ \cmd \},\sigma) & \text{if $\sem{\bexp}^s$}  \\
        f (\sigma) & \text{if $\sem{\neg\bexp}^s$;}
      \end{cases}
      \\
      & = \QETT[n]{f}(\cmd\sep\while(\bexp)\{ \cmd \},\sigma) +_{\sem{\bexp}^s} f(\sigma) \\
      & \cmt{Equation \eqref{l:QET:comp-approx:1}} \\
      & \leq \QETT[n]{\lambda \tau. \QETT[n]{f}(\while(\bexp)\{ \cmd \},\tau)}(\cmd,\sigma) +_{\sem{\bexp}^s} f(\sigma) \\
      & \cmt{induction hypothesis, monotonicity of $\QETT[n]{-}$} \\
      & \leq \QETT[n]{\chi_{f,\bexp}^n(\bot)}(\cmd,\sigma) +_{\sem{\bexp}^s} f(\sigma) \\
      & \leq \chi_{f,\bexp}^{n+1}(\bot)(\sigma)
        .
    \end{align*}
    \proofcase{$rhs \leq lhs$}
    We prove the stronger statement
    \[
      \chi_{f,\bexp}^n(\bot)(\sigma)
      \leq \QETT{f}(\while(\bexp)\{ \cmd \},\sigma)
      ,
    \]
    for all $n \in \N$ and $\sigma \in \memst$. The proof is by induction on $n$,
    where it is sufficient to consider the inductive step.
    Let $\sigma = (s,\ket{\varphi}) \in \memst$.
    Then
    \begin{align*}
      \fl{\chi_{f,\bexp}^{n+1}(\bot)(\sigma)} \\
      & = (\QETT{\chi_{f,\bexp}^n(\bot)}(\cmd,\sigma) +_{\sem{\bexp}^s} f(\sigma) ) \\
      & \cmt{induction hypothesis, monotonicity of $\QETT{-}$} \\
      & \leq \QETT{\lambda \tau. \QETT{f}(\while(\bexp)\{ \cmd \},\tau)}(\cmd,\sigma) +_{\sem{\bexp}^s} f(\sigma) \\
      & \cmt{equation \eqref{l:QET:comp-approx:2}, monotonicity of $\QETT{-}$} \\
      & \leq \QETT{f}(\cmd\sep\while(\bexp)\{ \cmd \},\sigma) +_{\sem{\bexp}^s} f(\sigma) \\
      & \cmt{reasoning as in the previous case} \\
      & = \QETT{f}(\while(\bexp)\{ \cmd \},\sigma) .
      \qedhere
    \end{align*}

  \end{proofcases}
\end{proof}

\section{Proofs of \Cref{thm:idents}}


\begin{lemma}\label{l:qt-continuity}
  For every
  $\omega$-chain $(f_n)_{n \in \N}$,
  \[
    \qet{\cmd}{\sup_{n\in\N} f_n} = \sup_{n \in \N} \qet{\cmd}{f_n}
    .
  \]
\end{lemma}
\begin{proof}
  The proof is by induction on the command $\cmd$.
  \begin{proofcases}
    \proofcase{$\skp$}
    Trivially,
    \[
      \qet{\skp}{\sup_{n\in\N} f_n} = \sup_{n\in\N} f_n = \sup_{n\in\N} \qet{\skp}{f_n}
      .
    \]
    \proofcase{$\x \asg \e$}
    Then
    \begin{align*}
      \fl{\qet{\x \asg \e}{\sup_{n\in\N} f_n}} \\
      & = \lambda \ltuple s,\ket{\varphi} \rtuple. \left(\sup_{n\in\N} f_n\right)\ltuple s[\x :=\sem{\e}^{s}],\ket{\varphi} \rtuple \\
      & = \lambda \ltuple s,\ket{\varphi} \rtuple. \sup_{n\in\N} f_n\ltuple s[\x :=\sem{\e}^{s}],\ket{\varphi} \rtuple\\
      & = \sup_{n\in\N} \lambda \ltuple s,\ket{\varphi} \rtuple. f_n\ltuple s[\x :=\sem{\e}^{s}],\ket{\varphi} \rtuple\\
      & = \sup_{n\in\N} \qet{\x \asg \e}{f_n} .
    \end{align*}
    \proofcase{$\x \asg \meas(\q)$}
    Then, similar as above,
    \begin{align*}
      \fl{\qet{\x \asg \meas(\q)}{\sup_{n\in\N} f_n}} \\
      & = \lambda \ltuple s,\ket{\varphi} \rtuple.
        \left(\sup_{n\in\N} f_n\right)\ltuple s[\x :=0],\oper{M}_0^\q\ket{\varphi} \rtuple\\
      & \qquad\quad +_{p_0^\q\ket{\varphi}} \left(\sup_{n\in\N} f_n\right)\ltuple s[\x :=1],\oper{M}_1^\q\ket{\varphi} \rtuple \\
      & = \lambda \ltuple s,\ket{\varphi} \rtuple.
        \sup_{n\in\N} f_n\ltuple s[\x :=0],\oper{M}_0^\q\ket{\varphi} \rtuple\\
      & \qquad\quad +_{p_0^\q\ket{\varphi}} \sup_{n\in\N} f_n\ltuple s[\x :=1],\oper{M}_1^\q\ket{\varphi} \rtuple \\
      & \cmt{continuity of barycentric sum} \\
      & = \lambda \ltuple s,\ket{\varphi} \rtuple.
        \sup_{n\in\N} \bigl(f_n\ltuple s[\x :=0],\oper{M}_0^\q\ket{\varphi} \rtuple\\
      & \qquad\qquad\qquad +_{p_0^\q\ket{\varphi}} f_n\ltuple s[\x :=1],\oper{M}_1^\q\ket{\varphi} \rtuple\bigr) \\
      & = \sup_{n\in\N} \lambda \ltuple s,\ket{\varphi} \rtuple.
        \bigl(f_n\ltuple s[\x :=0],\oper{M}_0^\q\ket{\varphi} \rtuple\\
      & \qquad\qquad\qquad +_{p_0^\q\ket{\varphi}} f_n\ltuple s[\x :=1],\oper{M}_1^\q\ket{\varphi} \rtuple\bigr) \\
      & = \sup_{n\in\N} \qet{\x \asg \meas(\q)}{f_n}
        .
    \end{align*}
    \proofcase{$\cons(\nexp)$}
    \begin{align*}
      \fl{\qet{\cons(\nexp)}{\sup_{n\in\N} f_n}} \\
      & = \lambda \ltuple s,\ket{\varphi} \rtuple. \max(\sem{\nexp}^s,0) \CSp \sup_{n\in\N} f_n \\
      & \cmt{continuity of $\CSp$}\\
      & = \lambda \ltuple s,\ket{\varphi} \rtuple. \sup_{n\in\N} (\max(\sem{\nexp}^s,0) \CSp f_n) \\
      & = \sup_{n\in\N} \lambda \ltuple s,\ket{\varphi} \rtuple. (\max(\sem{\nexp}^s,0) \CSp f_n) \\
      & = \sup_{n\in\N} \qet{\cons(\nexp)}{f_n}
        .
    \end{align*}
    \proofcase{$\cmd_1 \sep \cmd_2$}
    \begin{align*}
      \fl{\qet{\cmd_1\sep\cmd_2}{\sup_{n\in\N} f_n}}\\
      & = \qet{\cmd_1}{\qet{\cmd_2}{\sup_{n\in\N} f_n}} \\
      & \cmt{induction hypothesis on $\cmd_2$} \\
      & = \qet{\cmd_1}{\sup_{n\in\N} \qet{\cmd_2}{f_n}} \\
      & \cmt{induction hypothesis on $\cmd_1$} \\
      & = \sup_{n\in\N} \qet{\cmd_1}{\qet{\cmd_2}{f_n}} \\
      & = \sup_{n\in\N} \qet{\cmd_1\sep\cmd_2}{f_n}
        .
    \end{align*}

    \proofcase{$\ifa (\bexp) \{\cmd_1\} \elsea \{\cmd_2\}$}
    \begin{align*}
      \fl{\qet{\ifa (\bexp) \{\cmd_1\} \elsea \{\cmd_2\}}{\sup_{n \in \N} f_n}} \\
      & = \qet{\cmd_1}{\sup_{n \in \N} f_n} \up{\sem{\bexp}} \qet{\cmd_2}{\sup_{n \in \N} f_n} \\
      & \cmt{induction hypotheses} \\
      & = \sup_{n \in \N} \qet{\cmd_1}{f_n} \up{\sem{\bexp}} \sup_{n \in \N} \qet{\cmd_2}{f_n} \\
      & \cmt{continuity of barycentric sum} \\
      & = \sup_{n \in \N} \left(\qet{\cmd_1}{f_n} \up{\sem{\bexp}} \qet{\cmd_2}{f_n}\right) \\
      & = \sup_{n \in \N}\qet{\ifa (\bexp) \{\cmd_1\} \elsea \{\cmd_2\}}{f_n}
        .
    \end{align*}
    \proofcase{$\while(\bexp)\{ \cmd \}$}
    We conclude this final case as,
    \begin{align*}
      \fl{\qet{\while(\bexp)\{ \cmd \}}{\sup_{n \in \N} f_n}} \\
      & = \lfp\left(\lambda F. \qet{\cmd}{F} \up{\sem{\bexp}} \sup_{n \in \N} f_n\right) \\
      & \cmt{continuity of barycentric sum} \\
      & = \lfp\left(\lambda F. \sup_{n \in \N} \left(\qet{\cmd}{F} \up{\sem{\bexp}} f_n\right)\right) \\
      & = \lfp\left(\sup_{n \in \N} \lambda F. \left(\qet{\cmd}{F} \up{\sem{\bexp}} f_n\right)\right) \\
      & \cmt{$\star$} \\
      & = \sup_{n \in \N}\lfp\left(\lambda F.  \qet{\cmd}{F} \up{\sem{\bexp}} f_n\right) \\
      & = \sup_{n \in \N} \qet{\while(\bexp)\{ \cmd \}}{f_n}
      .
    \end{align*}
    Concerning $(\star)$ we use that $\lfp$ itself is continuous (\Cref{p:lfp}) on continuous functionals,
    the latter being a consequence of the induction hypothesis.
  \end{proofcases}
\end{proof}

\begin{lemma}[Monotonicity Law]\label{l:qt-monotone}
  \[
    f \leq g \implies \qet{\cmd}{f} \leq \qet{\cmd}{g}
    .
  \]
\end{lemma}
\begin{proof}
  This is an immediate consequence of \Cref{l:qt-continuity}.
\end{proof}

\begin{lemma}[Distributivity Law]\label{l:qt-distributivity}
  For all $p \in [0,1]$,
  \[
    \qet{\cmd}{f \up{p} g} = \qet{\cmd}{f} \up{p} \qet{\cmd}{g} .
  \]
\end{lemma}
\begin{proof}
  We reason semantically.
  Let $p \in [0,1]$ be a constant. Fix $\sigma\in \memst$. We have
  \begin{align*}
    \fl{\qet{\cmd}{f \up{p} g}(\sigma)}\\
    & \cmt{\Cref{t:soundness} and \Cref{l:QET}} \\
    & = \ecost_\cmd\ltuple s,\ket{\varphi} \rtuple \CSp \evalue_\cmd(f \up{p} g) (\sigma) \\
    & \cmt{\Cref{l:dsum-bari}, unfolding $\evalue_\cmd$} \\
    & = (p \cdot \ecost_\cmd (\sigma) + (1-p) \cdot \ecost_\cmd (\sigma))\\
    & \qquad \CSp (\evalue_\cmd(f) (\sigma) \up{p} \evalue_\cmd(g) (\sigma))\\
    & \cmt{\Cref{d:cs}(3)} \\
    & = (\ecost_\cmd (\sigma) \CSp (\evalue_\cmd(f) (\sigma))\\
    & \qquad \up{p} (\ecost_\cmd (\sigma) \CSp \evalue_\cmd(g)(\sigma)) \\
    & \cmt{\Cref{t:soundness} and \Cref{l:QET}} \\
    & = \qet{\cmd}{f}(\sigma) +_p \qet{\cmd}{g}(\sigma)
      .\qedhere
  \end{align*}
\end{proof}

\begin{lemma}[Upper Invariants]\label{l:qt-ui}
  If $\sem{\neg \bexp} \cdot f \leq g\ \wedge\ \sem{\bexp} \cdot \qet{\cmd}{g} \leq g$
  then
  \[
    \qet{\while(\bexp)\{\cmd\}}{f} \leq g
    .
  \]
\end{lemma}
\begin{proof}
  Let $\chi_f = \lambda F.\qet{\cmd}{F} \up{\sem{\bexp}} f$, hence
  \[
    \qet{\while(\bexp)\{\cmd\}}{f} = \lfp(\chi_f)
    .
  \]
  The hypothesis yields $\chi_f(g) \leq g$.
  As the least-fixed point of any functional is bounded by any such prefix-point, in particular $g$,
  the lemma follows.
\end{proof}

\section{Cost transformer laws}
In \Cref{fig:ctlaws}, we exhibit cost transformer laws that that make reasoning about quantum cost expectations easier.

The \ref{idents:rext-sep} law allows one to reason independently about the cost, $\qect{\cmd}{\underline{0}}$, and expectation of $f$, $\qevr{\cmd}{f}$.
 It enables a form of modular reasoning, e.g., when reasoning about the expected cost
 \[
   \qect{\cmd_1\sep\cmd_2}{\underline{0}} = \qect{\cmd_1}{\qect{\cmd_2}{\underline{0}}},
 \]
 of sequentially executed commands,
 the law states that it is sufficient bind the cost of $\cmd_1$ and $\cmd_2$ separately,
 and then investigate how $\cmd_1$ changes the latter in expectation.
 This can then be also combined with upper-invariants to reason about the cost of loops inductively (see \cite{AMS20}).
The \ref{idents:rext-lin} Law is derived from ``linearity of expectations''. Concerning the cost transformer, it is in general an inequality
because the cost is accounted twice in the right-hand side.

The \ref{idents:constancy} Law is inspired by the simple but useful,
equally named rule in Hoare logic (for classical programs),
stating that additional assumptions $P$ on initial states can be pushed to final states,
as long as $P$ is independent on the memory modified by the program fragment $\cmd$
under consideration, in notation $P \perp \cmd$. In the case of
weakest precondition calculi, this reads as $wp[\cmd]\{P \land Q\} = P \land wp[\cmd]\{Q\}$,
provided $P \perp \cmd$.
Moving from predicates $P$ to cost functions $f$, where conjunction is naturally
interpreted as multiplication, gives then rise to our law of constancy.
The notation $f \perp \cmd$ means that the value of $f$ remains unchanged
during evaluation of $\cmd$.
Syntactically, this property can be ensured by requiring that the expectation $f$ is constant in the variables in $\Bool(\cmd) \cup \Var(\cmd)$ assigned by $\cmd$, and in the qubits in $\Qubits(\cmd)$ measured within $\cmd$.
E.g., for $f = \lambda \ltuple s, \_ \rtuple$, $f \perp \y \asg \x + 3$ but
$f \perp \x \asg \x + 3$.
More precisely, $f \perp \cmd$ holds if  for all $\ltuple s,\ket{\varphi} \rtuple \in \memst$,
(i)~$f\ltuple s[\x^\K := a],\ket{\varphi} \rtuple = f$
for all $\x \in Var$ and $a \in \sem{\K}$, and
(ii)~$f\ltuple s[\x : = 0],t_0^\q\ket{\varphi} \rtuple = f = f\ltuple s[\x : = 1],t_1^\q\ket{\varphi} \rtuple$
for all $\q \in Qubit$.
Note that $f \perp \cmd$ holds in particular when $f$ is constant.

\ref{idents:constprop} Law falls in the same line of reasoning, and generalises
an equally named law from \cite{KKMO16}. Here,
the factor $\qevr{\cmd}{\underline{1}}$ to $f$ gives the termination probability of $\cmd$ ---
Hence, $\qect{\cmd}{f + g} = f + \qect{\cmd}{g}$ holds for $f \perp \cmd$, when $\cmd$ is almost-surely terminating.

We can show a result similar to \Cref{thm:idents}.
\begin{theorem}\label{thm:ctlaw}
  All cost transformer laws listed in \Cref{fig:idents} hold.
\end{theorem}
We prove this result in the remainder of this appendix.

\begin{figure*}[t]
\hrulefill
\begin{align*}
\mparbox{3mm}{\qevr{\cdot}{\cdot} : \Progs \to (\Rext)^\memst \to (\Rext)^\memst}\\
\mparbox{3mm}{\qect{\cdot}{\cdot} : \Progs \to (\Rext)^\memst \to (\Rext)^\memst}\\
 \quad
 & \law[idents:rext-sep]{separation}            &  & \qect{\cmd}{f} = \qect{\cmd}{\underline{0}} \up{} \qevr{\cmd}{f}                                                                \\
 & \law[idents:rext-lin]{linearity}             &  & \qect{\cmd}{f \up{} g} = \qect{\cmd}{f} \up{} \qevr{\cmd}{g} \leq \qect{\cmd}{f} \up{} \qect{\cmd}{g}                                   \\
 & \law[idents:constancy]{constancy}            &  & f \perp \cmd \implies \qect{\cmd}{f \ucdot g} = \qect{\cmd}{\underline{0}} \up{} f \ucdot \qevr{\cmd}{g} \leq \min(\underline{1},f) \ucdot \qect{\cmd}{g}\\
 & \law[idents:constprop]{constant propagation} &  & f \perp \cmd \implies \qect{\cmd}{f \up{} g} = \qevr{\cmd}{\underline{1}} \ucdot f + \qect{\cmd}{g} \leq f \up{} \qect{\cmd}{g}
\end{align*}
\hrulefill
\caption{Cost transformer laws.}
\label{fig:ctlaws}
\end{figure*}

\begin{lemma}[Separation]\label{l:qt-separation}
  \[
    \qect{\cmd}{f} = \qect{\cmd}{\underline{0}} \up{} \qevr{\cmd}{f}
  \]
\end{lemma}
\begin{proof}
  By \Cref{t:soundness} and \Cref{l:QET},
  \[
    \qect{\cmd}{f} = \ecost_\cmd \up{} \evalue_\cmd(f)
    ,
  \]
  the lemma then follows from \Cref{c:ect-soundness}.
\end{proof}

\begin{lemma}[Linearity]\label{l:qt-linearity}
  \begin{align*}
    \qect{\cmd}{f \up{} g} & = \qect{\cmd}{f} \up{} \qevr{\cmd}{g} \\
                       & \leq \qect{\cmd}{f} \up{} \qect{\cmd}{g}
                         .
  \end{align*}
\end{lemma}
\begin{proof}
An immediate consequence of \Cref{t:soundness} and \Cref{l:QET}
is that $\qevr{\cmd}{g} \leq \qevr{\cmd}{f}$. It is thus sufficient
to verify only the equality. We again proceed semantically:
  \begin{align*}
    \fl{\qect{\cmd}{f \up{} g}(\sigma)}\\
    & \cmt{\Cref{t:soundness} and \Cref{l:QET}} \\
    & = \ecost_\cmd(\sigma) + \evalue_\cmd(f \up{} g) (\sigma) \\
    & \cmt{linearity of expectations} \\
    & = \ecost_\cmd(\sigma) + \evalue_\cmd(f) (\sigma) + \evalue_\cmd(g) (\sigma)\\
    & \cmt{\Cref{t:soundness} and \Cref{l:QET}} \\
    & = \qect{\cmd}{f}(\sigma) + \evalue_\cmd(g) (\sigma) \\
    & \cmt{\Cref{c:ect-soundness}} \\
    & = \qect{\cmd}{f}(\sigma)  + \qevr{\cmd}{g}
      .\qedhere
  \end{align*}
\end{proof}

\begin{lemma}[Constancy]\label{l:qt-constancy}
  Suppose $f \perp \cmd$. Then
  \begin{enumerate}
  \item\label{l:qt-constancy:qev} $\qevr{\cmd}{f \ucdot g} = f \cdot \qevr{\cmd}{f \ucdot g}$;
  \item\label{l:qt-constancy:qect} $\qect{\cmd}{f \ucdot g} = \qect{\cmd}{\underline{0}} \up{} f \ucdot \qevr{\cmd}{g}$\\
    ${} \qquad \leq \umin(\underline{1},f) \ucdot \qect{\cmd}{g}$.
  \end{enumerate}
\end{lemma}
\begin{proof}
  Suppose $f \perp \cmd$. Note that~\eqref{l:qt-constancy:qect} is a consequence of \eqref{l:qt-constancy:qev}:
  \begin{align*}
    \fl{\qect{\cmd}{f \ucdot g}} \\
    & \cmt{\Cref{l:qt-separation}} \\
    & = \qect{\cmd}{\underline{0}} \up{} \qevr{\cmd}{f \ucdot g} \\
    & \cmt{identity \eqref{l:qt-constancy:qev}} \\
    & = \qect{\cmd}{\underline{0}} \up{} f \ucdot \qevr{\cmd}{g} \\
    & \leq \umin(\underline{1},f) \ucdot (\qect{\cmd}{\underline{0}} \up{} \qevr{\cmd}{g}) \\
    & \cmt{\Cref{l:qt-separation}} \\
    & = \umin(\underline{1},f) \ucdot \qect{\cmd}{g}
  \end{align*}

  It is thus sufficient to prove \eqref{l:qt-constancy:qev},
  which we do by  induction on $\cmd$.
  For $f,g,h : (\Rext)^\memst$ and $p : [0,1]^\memst$,
  we will employ the identity
  \begin{align*}
    f \ucdot g \up{p} f \ucdot h
    & = p \ucdot f \ucdot g \up{} (\underline{1} \ud p) \ucdot f \ucdot h \\
    & = f \ucdot (p \ucdot g \up{} (\underline{1} \ud p) \ucdot h) \\
    & = f \ucdot (g \up{p} h)
      ,
  \end{align*}
  in several cases.
  \begin{proofcases}
    \proofcase{$\skp$}
    The case is trivial.
    \proofcase{$\x \asg \e$}
    As $f \perp (\x \asg \e)$ we have $f[\x := \e] = f$.
    Thus
    \begin{align*}
      \qevr{\x \asg \e}{f \ucdot g}
      & = (f \ucdot g)[\x := \e] \\
      & = f \ucdot g[\x := \e] \\
      & = f \ucdot \qevr{\x \asg \e}{g}
        .
    \end{align*}
       \proofcase{$\x \asg \meas(\q)$}
    In this case we have $f[\x := 0;\oper{M}_0^\q] = f = f[\x := 1;\oper{M}_1^\q]$.
    We conclude then as
    \begin{align*}
      \fl{\qevr{\x \asg \meas(\q)}{f \ucdot g}} \\
      & = (f \ucdot g)[\x := 0;\oper{M}_0^\q] \up{p_0^\q} (f \ucdot g)[\x := 1;\oper{M}_1^\q] \\
      & = f \ucdot g[\x := 0;\oper{M}_0^\q] \up{p_0^\q} f \ucdot g[\x := 1;\oper{M}_1^\q] \\
      & = f \ucdot (g[\x := 0;\oper{M}_0^\q] \up{p_0^\q} g[\x := 1;\oper{M}_1^\q]) \\
      & = f \ucdot \fl{\qevr{\x \asg \meas(\q)}{g}}
        .
    \end{align*}
    \proofcase{$\cons(\nexp)$}
    Then
    \begin{align*}
      \fl{\qevr{\cons(\nexp)}{f \cdot g}} \\
      & = \umax(\sem{a},\underline{0}) \up{\mathsf{f}} f \ucdot g \\
      & = f \ucdot g \\
      & = f \ucdot \qevr{\cons(\nexp)}{g} .
    \end{align*}
    \proofcase{$\cmd_1 \sep \cmd_2$}
    Note that also $f \perp \cmd_1$ and $f \perp \cmd_2$, by assumption.
    Thus we can conclude via induction hypothesis:
    \begin{align*}
      \fl{\qevr{\cmd_1 \sep \cmd_2}{f \ucdot g}}\\
      & = \qevr{\cmd_1}{\qevr{\cmd_2}{f \ucdot g}} \\
      & \cmt{induction hypothesis on $\cmd_2$} \\
      & = \qevr{\cmd_1}{f \ucdot \qevr{\cmd_2}{g}} \\
      & \cmt{induction hypothesis on $\cmd_1$} \\
      & = f \ucdot \qevr{\cmd_1}{\qevr{\cmd_2}{g}}
        .
    \end{align*}
    \proofcase{$\ifa (\bexp) \{\cmd_1\} \elsea \{\cmd_2\}$}
    Again, $f \perp \cmd_1$ and $f \perp \cmd_2$, by assumption,
    and hence via induction hypothesis:
    \begin{align*}
      \fl{\qevr{\ifa (\bexp) \{\cmd_1\} \elsea \{\cmd_2\}}{f \ucdot g}}\\
      & = \qevr{\cmd_1}{f \ucdot g} \up{\sem{\bexp}} \qevr{\cmd_2}{f \ucdot g} \\
      & \cmt{induction hypotheses} \\
      & = f \ucdot \qevr{\cmd_1}{g} \up{\sem{\bexp}} f \cdot \qevr{\cmd_2}{g} \\
      & = f \ucdot (\qevr{\cmd_1}{g} \up{\sem{\bexp}} \qevr{\cmd_2}{g}) \\
      & = f \ucdot \qevr{\ifa (\bexp) \{\cmd_1\} \elsea \{\cmd_2\}}{g}
        .
    \end{align*}
    \proofcase{$\while(\bexp)\{ \cmd \}$}
    Let $\chi_h = \lambda F.\qet{\cmd}{F} \up{\sem{\bexp}} h$, hence
    for any $h$,
    \[
      \qet{\while(\bexp)\{\cmd\}}{h} = \lfp(\chi_h) = \sup_n \chi_h^n(\bot)
      .
    \]
    It is thus sufficient to prove
    \[
      \forall n \in \N. \chi_{f+g}^n(\bot) = f \ucdot \chi_g^n(\bot) .
    \]
    The proof is by induction on $n$. In the base case $n = 0$,
    notice $\bot = \underline{0}$ and trivially $\underline{0} = f \ucdot \underline{0}$.
    Hence consider the inductive step. By assumption also $f \perp \cmd$ holds.
    Thus
    \begin{align*}
      \chi_{f\up{}g}^{n+1}(\bot)
      & = \qet{\cmd}{\chi_{f+g}^{n}(\bot)} \up{\sem{\bexp}} f \ucdot g \\
      & \cmt{side induction hypothesis} \\
      & = \qet{\cmd}{f \cdot \chi_{g}^{n}(\bot)} \up{\sem{\bexp}} f \ucdot g \\
      & \cmt{induction hypothesis on $\cmd$} \\
      & = f \ucdot \qet{\cmd}{\chi_{g}^{n}(\bot)} \up{\sem{\bexp}} f \ucdot g \\
      & = f \ucdot (\qet{\cmd}{\chi_{g}^{n}(\bot)} \up{\sem{\bexp}} g) \\
      & = f \ucdot \chi_{g}^{n+1}(\bot)
        .
        \qedhere
    \end{align*}
  \end{proofcases}
\end{proof}

\begin{lemma}[Constant Propagation]\label{l:qt-constprop}
  Suppose $f \perp \cmd$.
  Then
  \begin{align*}
    \qect{\cmd}{f + g}
    & = \qevr{\cmd}{\underline{1}} \cdot f + \qect{\cmd}{g} \\
    & \leq f + \qect{\cmd}{g}
      .
  \end{align*}
\end{lemma}
\begin{proof}
  Note $\qevr{\cmd}{\underline{1}} \leq \underline{1}$ for the termination
  probability $\qevr{\cmd}{\underline{1}}$. It is thus sufficient to prove
  the equality:
  \begin{align*}
    \fl{\qect{\cmd}{f \up{} g}} \\
    & \cmt{\Cref{l:qt-separation}} \\
    & = \qevr{\cmd}{f} \up{} \qect{\cmd}{g} \\
    & = \qevr{\cmd}{f \cdot \underline{1}} \up{} \qect{\cmd}{g} \\
    & \cmt{\Cref{l:qt-constancy}} \\
    & = f \cdot \qevr{\cmd}{\underline{1}} \up{} \qect{\cmd}{g}
      .
      \qedhere
  \end{align*}
\end{proof}

\end{document}